\documentclass[11pt, sort]{article}

\usepackage{acl}

\usepackage{times}
\usepackage{latexsym}

\usepackage[T1]{fontenc}
\usepackage[utf8]{inputenc}
\usepackage{amsmath}
\usepackage{amsthm}
\newtheorem{theorem}{Theorem}
\newtheorem{proposition}{Proposition}

\usepackage{microtype}

\usepackage{inconsolata}

\usepackage{graphicx}

\usepackage{booktabs}
\usepackage{multirow}
\usepackage{amssymb}
\usepackage{enumitem}

\usepackage[table]{xcolor}

\usepackage{xltabular} % 支持自动换行 + 跨页
\usepackage{ragged2e}  % 用于对齐优化

\usepackage{tabularx}
\usepackage{array}
\usepackage{bbm}

\usepackage{float}

\title{MASH: Evading Black-Box AI-Generated Text Detectors \\via Style Humanization}

\author{
    Yongtong Gu\textsuperscript{1}, 
    Songze Li\textsuperscript{1}\thanks{Corresponding author.}, 
    Xia Hu\textsuperscript{2} \\
    \vspace{0.5em} 
    \textsuperscript{1}Southeast University \\[-0.2cm] 
    \textsuperscript{2}Shanghai Artificial Intelligence Laboratory \\
    \vspace{0.5em} 
    \texttt{\{yongtonggu, songzeli\}@seu.edu.cn}, \texttt{huxia@pjlab.org.cn}
}

\begin{document}
\maketitle
\begin{abstract}

The increasing misuse of AI-generated texts (AIGT) has motivated the rapid development of AIGT detection methods.
However, the reliability of these detectors remains fragile against adversarial evasions.
Existing attack strategies often rely on white-box assumptions or demand prohibitively high computational and interaction costs, rendering them ineffective under practical black-box scenarios.
In this paper, we propose Multi-stage Alignment for Style Humanization (MASH), a novel framework that evades black-box detectors based on style transfer.
MASH sequentially employs style-injection supervised fine-tuning, direct preference optimization, and inference-time refinement to shape the distributions of AI-generated texts to resemble those of human-written texts.
Experiments across 6 datasets and 5 detectors demonstrate the superior performance of MASH over 11 baseline evaders.
Specifically, MASH achieves an average Attack Success Rate (ASR) of 92\%, surpassing the strongest baselines by an average of 24\%, while maintaining superior linguistic quality.
Our code and data are publicly available at \url{https://github.com/githigher/MASH}.
% Specifically, MASH achieves an average ASR of 92\% against commercial detectors and surpasses the strongest baselines by an average of 24\% on open-source detectors, while maintaining superior linguistic quality.

\end{abstract}

\section{Introduction}

Large Language Models (LLMs), such as ChatGPT, Gemini, and Claude, have demonstrated remarkable capabilities across various domains, including social media \citep{macko2025multisocial}, creative writing \citep{baek2025researchagent}, and code generation \citep{li2025large}. 
However, the rapid proliferation of these capabilities has raised serious concerns across academia and industry over issues such as automated fake news, web content pollution, and academic plagiarism \citep{wang2024stumbling, zheng2025th}.
To maintain the authenticity of the information ecosystem, detecting AI-generated text (AIGT) has become a critical issue.
% Current detection methodologies include watermark-based detection mechanisms \citep{zhao2023provable}, supervised fine-tuned classifiers \citep{li2024mage, huang-etal-2024-ai}, and zero-shot detection methods based on statistical features \citep{hans2024spotting}.
% Recent studies indicate that these detectors have achieved remarkable accuracy on standard benchmarks \citep{russell-etal-2025-people, cheng2025beyond}.
Recent studies \citep{russell-etal-2025-people, cheng2025beyond} indicate that existing AIGT detectors have achieved remarkable accuracy on standard benchmarks.

\begin{figure}[t]
\includegraphics[width=\columnwidth]{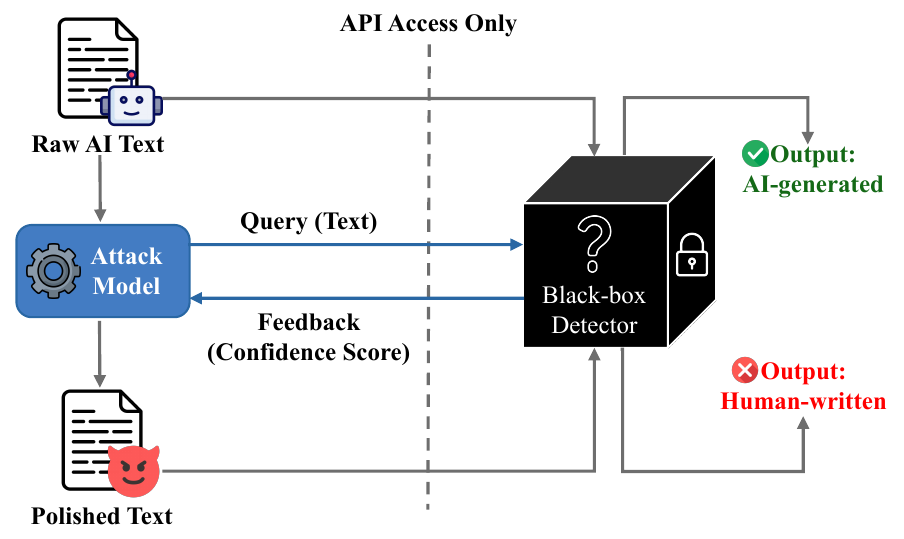}
\caption{Illustration of an evasion attack against AIGT detectors in a black-box setting.}
\label{fig:attack_illustration}
\end{figure}

Nevertheless, the robustness of these detectors remains fragile when facing malicious attacks. As illustrated in Figure \ref{fig:attack_illustration}, an attacker can rewrite the original generation to alter its stylistic pattern, thereby fooling the detector into misclassifying it as human-written.
% Current attack strategies are primarily categorized into three types \citep{wang2024stumbling}: perturbation-based, prompt-based, and paraphrase-based attacks.
% However, these methods neglect the deep semantic feature distributions between AI-generated and human-written text, failing to deceive robust detectors \citep{krishna2023paraphrasing}.
% % Furthermore, these methods lack practical feasibility due to several limitations: 
% Furthermore, these methods lack practical feasibility in resource-constrained black-box scenarios due to several limitations:
% (1) the need for many interactions with detectors during inference \citep{Zhou2024_COLING, Abad2024Charmer, jin2020bert, li2018textbugger, gao2018black}; (2) high computational costs for training or inference \citep{fang-etal-2025-language}; and (3) the requirement for access to detector gradients \citep{meng2025gradescape}. 
% % and (4) the inevitable degradation of text fluency.
However, existing attack strategies often rely on white-box assumptions, leading to limited Attack Success Rates (ASR) in black-box settings. Furthermore, these methods face significant obstacles in real-world deployment, where detectors typically operate as black boxes with strict access limits. Major constraints include: (1) excessive query interactions during inference \citep{Zhou2024_COLING, Abad2024Charmer}; (2) high computational burdens \citep{fang-etal-2025-language}; and (3) unrealistic requirements for detector gradients \citep{meng2025gradescape}.

% The robustness of current detectors exploits the inherent feature discrepancies between AI-generated and human-written text, resulting in a clear separation in both semantic \citep{huang-etal-2024-ai} and statistical domains \citep{hans2024spotting}. Consequently, we posit that effective evasion requires feature distribution alignment, rather than adversarial noise injection.
% However, directly applying alignment algorithms often fails due to the ``cold start'' dilemma, where the generator's initial lack of support for human stylistic features leads to optimization collapse. 
% To resolve this, MASH sequentially bridges the gap: it first employs Style-injection Supervised Fine-Tuning (Style-SFT) to explicitly pull the generation distribution toward the human manifold for a robust initialization; this enables effective Direct Preference Optimization (DPO) \citep{rafailov2023direct} to penetrate detector boundaries, culminating in an inference-time adversarial refinement to guarantee linguistic quality.

To address these challenges, we propose Multi-stage Alignment for Style Humanization (MASH), a novel black-box style transfer framework for detector evasion.
Current detectors identify AI-generated text by leveraging inherent distribution shifts in semantic \citep{huang-etal-2024-ai} and statistical features \citep{hans2024spotting}.
Consequently, we posit that effective evasion necessitates humanizing AI-specific features through style transfer. MASH trains a paraphraser via a sequential pipeline: First, it employs Style-injection Supervised Fine-Tuning (Style-SFT) on open-source corpora, where the model learns by explicitly extracting style vectors from AI-generated and human-written text, respectively. Building on this foundation, the paraphraser is further fine-tuned using Direct Preference Optimization (DPO) \citep{rafailov2023direct} to adapt to the detector's decision boundaries, thereby significantly boosting the ASR. Finally, adversarial refinement is applied during inference to guarantee the linguistic quality of the output.

In summary, our contributions are as follows:

\begin{itemize}
    \item We propose MASH, a multi-stage black-box attack framework that sequentially integrates Style-SFT, DPO alignment, and inference-time refinement to achieve effective evasion.
    \item We demonstrate that a 0.1B parameter model, optimized via MASH, can surpass much larger models in evasion performance while requiring only limited interactions.
    \item Systematic evaluations across six domains confirm that MASH significantly outperforms state-of-the-art baselines, achieving superior attack success rates against both open-source and commercial detectors while maintaining superior linguistic quality.
\end{itemize}

\section{Related Work}

\subsection{AI-generated Text Detection}

\noindent \textbf{Training-based Detectors.} These methods typically formulate detection as a binary classification problem. Recent studies have incorporated adversarial training \citep{mao2024raidar}, Longformer architectures \citep{li2024mage}, contrastive learning \citep{guo2024detective}, and denoising reconstruction mechanisms \citep{huang-etal-2024-ai}.
Furthermore, such supervised approaches typically require extensive manual annotations.

\noindent \textbf{Zero-shot Detectors.} These methods eliminate the need for training by leveraging inherent statistical features. Representative approaches include GLTR \citep{gehrmann-etal-2019-gltr}, which utilizes token ranking, and DetectGPT \citep{mitchell2023detectgpt}, which examines probability curve curvature. Other works include Fast-DetectGPT \citep{bao2023fast} and Binoculars \citep{hans2024spotting}; the latter identifies machine text through ratio relationships between perplexity and cross-perplexity.
This study focuses only on training-based and zero-shot detectors due to their practical prevalence in deployment.

% \noindent \textbf{Watermark-based detectors.} These methods introduce invisible patterns during generation for subsequent identification \citep{zhao2023provable, qu2025provably}.
% However, this study focuses on the first two categories, because they remain the most prevalent in practical deployment.

\subsection{Detection Evasion Methods}

\noindent \textbf{Perturbation-based Attacks.} These methods disrupt detectors through character or word-level modifications. Early works like DeepWordBug \citep{gao2018black} and TextFooler \citep{jin2020bert} utilize spelling errors or synonym substitutions. Recently, Charmer \citep{Abad2024Charmer} revitalizes this direction by optimizing character-level noise distributions. 
% Despite their effectiveness, these approaches often suffer from a significant decline in text fluency, limiting their practical utility.
However, their ASR remains suboptimal when the budget for detector interactions is strictly constrained.

\noindent \textbf{Prompt-based Attacks.} These methods use contextual prompts to guide LLMs in generating evasion-prone text. SICO \citep{lu2024large} employs optimization to find detection-minimizing prompts, while PromptAttack \citep{xu2024an} uses prompts to mimic human styles. Despite producing high-quality text, their efficacy depends on the model's instruction-following consistency and can be unstable against robust classifiers.

\noindent \textbf{Paraphrase-based Attacks.} These methods utilize specialized models to rephrase machine-generated text. DIPPER \citep{krishna2023paraphrasing} introduces a T5-XXL paraphraser that enables control over lexical and syntactic parameters. Subsequently, \citet{sadasivan2023can} proposed using recursive paraphrasing. 
% Nevertheless, these rewriting models fail to capture the core essence of the evasion task.
Recent works leverage DPO for evasion, using detector confidence as a reward \citep{nicks2023language, wang2025humanizing}. Building on this, \citet{pedrotti-etal-2025-stress} fine-tune models on diverse parallel data. However, these methods often necessitate white-box access to the source generator. In contrast, MASH is a black-box paradigm that humanizes text from any LLM without internal access to either the source generator or the target detector.

\subsection{Text Style Transfer}

% Text style transfer (TST) aims to modify linguistic attributes such as sentiment, while strictly preserving semantic fidelity. Early methodologies utilized latent disentanglement \cite{hu2017toward} or adversarial alignment \cite{shen2017style} to decouple style from content. The advent of LLMs has shifted the paradigm toward in-context learning and parameter-efficient fine-tuning \cite{reif2022recipe}.

Text style transfer (TST) aims to modify linguistic attributes while strictly preserving semantic fidelity, such as sentiment or formality, via learning disentangled latent representations with designated semantics \citep{hu2017toward, shen2017style}. Recently, LLMs have shifted this paradigm toward in-context learning and parameter-efficient fine-tuning \citep{reif2022recipe}. 
Notably, recent studies reveal the intrinsic stylistic divergence between AI-generated and human-written text \citep{liu2023chatgpt, jiang2025sendetex}.
In this work, we extend TST to stylistic humanization. redefining detector evasion as a specialized ``machine-to-human'' style transfer task.

\section{Methodology}

\subsection{Problem Formulation}

% \noindent %\textbf{Machine-Generated Text Detection.} 
The task of AI-generated text detection is formulated as a binary classification problem. Given a sequence $\mathbf{x} = (x_1, \dots, x_T)$ generated by an LLM via $P_{LLM}(\mathbf{x}) = \prod_{t=1}^{T} P(x_t | x_{<t})$, a detector $D$ assigns a score $D(\mathbf{x}) \in [0, 1]$ representing the probability of AI. The final decision $\hat{y}$ is:
\begin{equation}
\hat{y} = \mathbbm{1}(D(\mathbf{x}) > \tau),
\end{equation}
where $\mathbbm{1}(\cdot)$ is the indicator function, $\tau$ is the threshold, and $\hat{y} \in \{1, 0\}$ denotes AI-generated and human-written text, respectively.

\noindent \textbf{Threat Model.} We consider a black-box setting where the attacker has no knowledge of the detector's architecture, parameters, or gradients. Access is restricted to an Oracle providing scores $D(\mathbf{x})$ or labels $\hat{y}$. The objective is to learn a generator $G_\phi$ that transforms source text $\mathbf{x}_{ai}$ into an adversarial example $\mathbf{x}_{adv}$ to evade detection:
\begin{equation}
\begin{aligned}
& \text{minimize} \quad \mathbb{E}_{\mathbf{x}_{ai} \sim P_{LLM}} \left[ D(G_\phi(\mathbf{x}_{ai})) \right], \\
& \text{subject to} \quad \mathcal{S}(\mathbf{x}_{adv}, \mathbf{x}_{ai}) \ge \epsilon, \quad \mathcal{Q}(\mathbf{x}_{adv}) \ge \delta,
\end{aligned}
\end{equation}
where $\mathcal{S}$ and $\mathcal{Q}$ denote semantic consistency and fluency with thresholds $\epsilon$ and $\delta$, respectively.

\begin{figure*}[t]
  \centering  % 图片居中
  \includegraphics[width=1\linewidth]{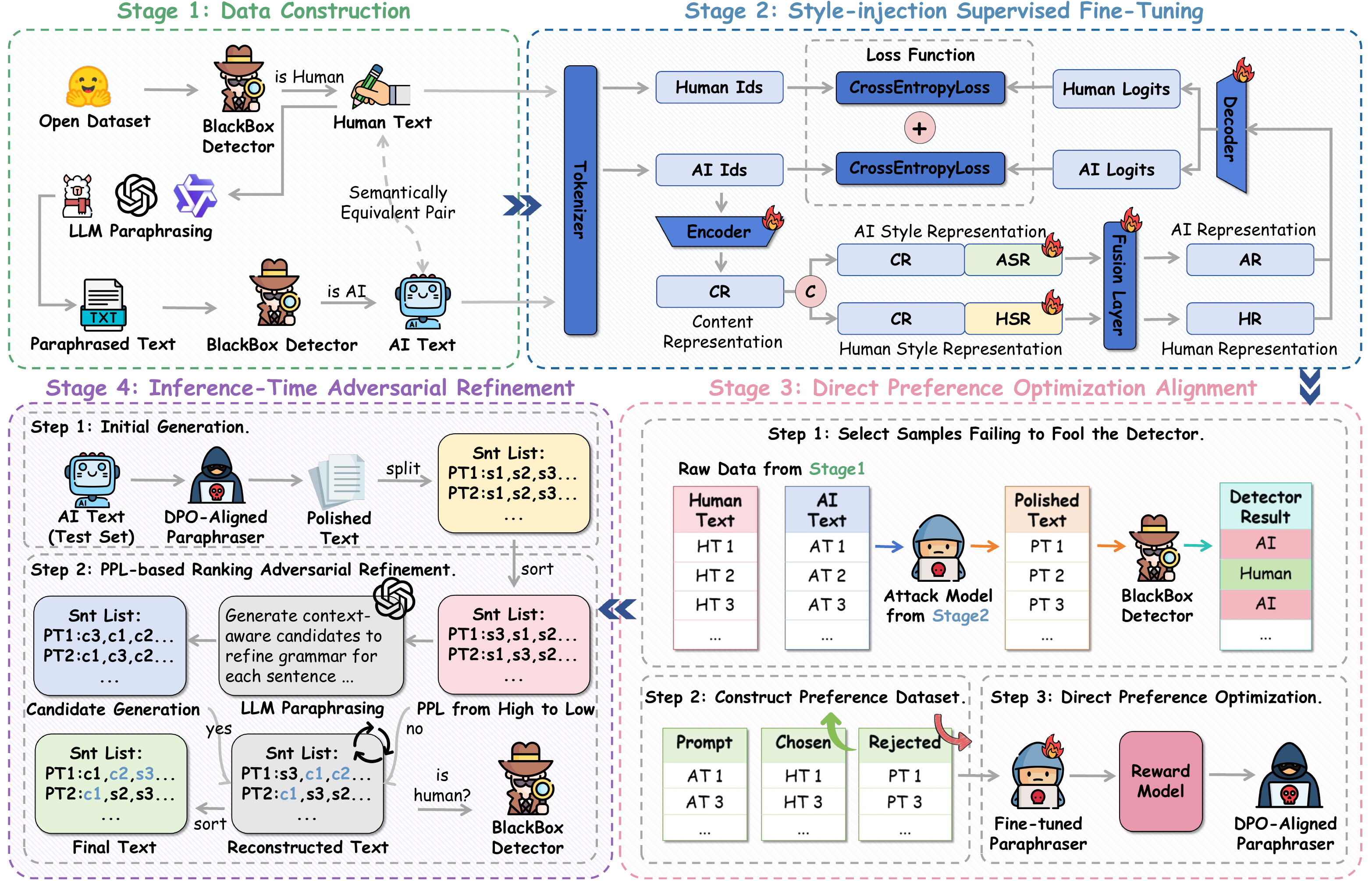} 
  \caption{Overview of the proposed MASH framework. It comprises four stages: (1) Data Construction to synthesize parallel data; (2) Style-Injection SFT for supervised initialization; (3) DPO Alignment to optimize against detector boundaries; and (4) Inference-Time Adversarial Refinement to guarantee final text quality.}
  \label{fig:mash} %以此标签引用图片
\end{figure*}

% \subsection{Overview}
To achieve the adversarial objective in a black-box setting, we develop the Multi-stage Alignment for Style Humanization (MASH) framework.
Our key idea is to reformulate detector evasion as a style transfer task.
As illustrated in Figure~\ref{fig:mash}, MASH bridges the stylistic gap between AI and human texts through four stages. We first synthesize parallel data via Inverse Data Construction, enabling Style-SFT to initialize the paraphraser with human writing patterns. Subsequently, DPO Alignment optimizes the paraphraser against detector boundaries, culminating in Inference-Time Adversarial Refinement to ensure linguistic quality.

\subsection{MASH Methodology}
\label{sec:mash_methodology}

\noindent \textbf{Stage 1: Data Construction.}
% We formulate evasion as a controlled style transfer task.
% To address the scarcity of text pairs, we leverage a key generation asymmetry: while removing machine fingerprints is challenging, the reverse—generating AI text with statistical artifacts from human sources—is trivial.
% Based on this, we design an Inverse Data Construction pipeline. 
Text style transfer is often limited by the scarcity of parallel training data. Our key insight is that while converting AI-generated text to human style is challenging, the reverse direction is easy: LLMs can readily rewrite human-written text into machine style. Moreover, human-written text is abundantly available from open-source datasets.
First, we collect raw texts from open-source datasets and retain only high-confidence human-written samples $\mathbf{x}_{human}$ (where $D(\mathbf{x}_{human}) < \tau$).
% Next, we employ an LLM to paraphrase these samples into semantically equivalent AI-generated texts $\mathbf{x}_{ai}$.
Next, for each human sample $\mathbf{x}_{human}$, we use an LLM to generate a semantically equivalent paraphrase $\mathbf{x}_{ai}$. These are further filtered to ensure strong machine style ($D(\mathbf{x}_{ai}) > \tau$). 
Consequently, we obtain a dataset consisting of $N$ AI-human pairs $\mathcal{D}_{pair} = \{(\mathbf{x}_{ai}^{(i)}, \mathbf{x}_{human}^{(i)})\}_{i=1}^N$.

\noindent \textbf{Stage 2: Style-injection Supervised Fine-Tuning.}
Direct fine-tuning often leads to overfitting and semantic loss. To achieve controllable humanization, we propose a style-injection architecture based on pre-trained BART. We decouple style by introducing trainable embeddings $\mathbf{s}_{ai}, \mathbf{s}_{human} \in \mathbb{R}^d$. Given input $\mathbf{x}_{ai}$, the encoder produces content representation $\mathbf{H}_{content}$. A fusion layer then injects the selected style $\mathbf{s}_{style}$ via linear projection:
\begin{equation}
\mathbf{H}_{fused}^{(t)} = \mathbf{W}_p \cdot [\mathbf{h}_{content}^{(t)} ; \mathbf{s}_{style}] + \mathbf{b}_p,
\end{equation}
where $[\cdot ; \cdot]$ denotes the concatenation, and $\mathbf{W}_p$ and $\mathbf{b}_p$ represent the trainable weight matrix and bias vector of the projection layer, respectively.

% The decoder then generates the target sequence conditioned on the fused representation $\mathbf{H}_{fused}$. To balance style intensity with semantic preservation, we employ a multi-task objective:

% \textbf{Style Reconstruction.} To ensure semantic stability and prevent catastrophic forgetting, we regularize the model by reconstructing the input $\mathbf{x}_{ai}$ conditioned on machine style $\mathbf{s}_{ai}$:
% \begin{equation}
% \mathcal{L}_{recon} = -\sum \log P_{\theta}(\mathbf{x}_{ai} | \mathbf{x}_{ai}, \mathbf{s}_{ai}),
% \end{equation}
% where $P_{\theta}$ denotes the paraphraser's likelihood.

% \textbf{Style Transfer.} To achieve evasion, we supervise the model to map AI-generated text to human target $\mathbf{x}_{human}$ conditioned on human style $\mathbf{s}_{human}$:
% \begin{equation}
% \mathcal{L}_{trans} = -\sum \log P_{\theta}(\mathbf{x}_{human} | \mathbf{x}_{ai}, \mathbf{s}_{human}).
% \end{equation}

% We minimize $\mathcal{L}_{SFT} = \lambda \mathcal{L}_{recon} + (1 - \lambda) \mathcal{L}_{trans}$ to inject human style while preserving semantics, ensuring robust initialization for alignment.

The decoder generates the target sequence conditioned on the fused representation $\mathbf{H}_{fused}$. To balance style injection with semantic preservation, we employ a multi-task objective:
\begin{equation}
\mathcal{L}_{recon} = -\sum \log P_{\theta}(\mathbf{x}_{ai} | \mathbf{x}_{ai}, \mathbf{s}_{ai}),
\end{equation}
\begin{equation}
\mathcal{L}_{trans} = -\sum \log P_{\theta}(\mathbf{x}_{human} | \mathbf{x}_{ai}, \mathbf{s}_{human}),
\end{equation}
where $\mathcal{L}_{recon}$ preserves semantic consistency by reconstructing the original AI-generated text, while $\mathcal{L}_{trans}$ imparts human stylistic patterns by generating the target $\mathbf{x}_{human}$ conditioned on $\mathbf{s}_{human}$.
We minimize $\mathcal{L}_{SFT} = \lambda \mathcal{L}_{recon} + (1 - \lambda) \mathcal{L}_{trans}$ to inject human style while preserving semantics, ensuring robust initialization for alignment.

\noindent \textbf{Stage 3: DPO Alignment.}
% While SFT provides preliminary humanization, it suffers from exposure bias and lacks a direct adversarial objective. To bridge this gap, we employ DPO to fine-tune the model against the detector's decision boundaries.
In the context of style transfer, SFT learns what human style looks like, but not how far the output has moved from AI style. To bridge this gap, we employ DPO, which uses the detector's confidence as an implicit reward signal to directly optimize for crossing the style boundary.
Formally, we posit that the human preference follows a Bradley-Terry model \citep{bradley1952rank} driven by an implicit reward function $r(x,y)$. To align the optimization objective with our evasion goal, we define the reward as a scaled inversion of the detector's confidence score $D(y)$:
\begin{equation}
    r(x,y) = -C \cdot D(y), \quad C > 0.
\label{eq:reward_def}
\end{equation}

As demonstrated in Appendix~\ref{sec:appendix_alignment}, maximizing this reward forces the optimal policy $\pi^*$ to converge towards regions where $D(y) \to 0$.
% yielding a closed-form solution to the adversarial objective.

\textbf{Hard Negative Mining.} 
We construct the preference dataset by repurposing the parallel data $\mathcal{D}_{pair}$ from Stage 1, eliminating the need for external data.
For each source input $\mathbf{x}_{ai}$, the ground-truth human text serves as the chosen response $\mathbf{y}_w$.
The rejected response $\mathbf{y}_l$ is sampled from the Stage 2 model $\pi_{\text{SFT}}(\cdot| \mathbf{x}_{ai}, \mathbf{s}_{human})$ and retained only if it fails to evade the detector (i.e., $D(\mathbf{y}_{l}) > \tau$).
As justified in Appendix~\ref{sec:appendix_hard_negative}, this selection criterion maximizes the probabilistic gap between pairs, preventing vanishing gradients and ensuring robust optimization.
The resulting dataset is formulated as:
\begin{equation}
\mathcal{D}_{\text{DPO}} = \left\{ (\mathbf{x}^{(i)}, \mathbf{y}_w^{(i)}, \mathbf{y}_l^{(i)}) \mid D(\mathbf{y}_{l}^{(i)}) > \tau \right\}_{i=1}^N.
\end{equation}

% We optimize $\pi_\theta$ to maximize the evasion reward under a KL constraint relative to the reference model $\pi_{\text{ref}}$ (i.e., $\pi_{\text{SFT}}$).
% The theoretical optimal policy takes the form:
% \begin{equation}
% \pi^*(\mathbf{y} | \mathbf{x}) \propto \pi_{\text{ref}}(\mathbf{y} | \mathbf{x}) \exp \left( \frac{r^*(\mathbf{x}, \mathbf{y})}{\beta} \right).
% \end{equation}

% This multiplicative dependency highlights the support constraint analyzed in Appendix~\ref{sec:appendix_initialization}: $\pi^*$ is strictly bounded by the support of $\pi_{\text{ref}}$.
% Thus, Stage 2 SFT is mathematically requisite to ensure $\pi_{\text{ref}}$ covers the target human manifold, preventing the ``cold start'' problem where $\pi_{\text{ref}}(\mathbf{y}_w | \mathbf{x}) \approx 0$.
\textbf{Fine-Tuning with DPO.}
The theoretical justification for using Style-SFT to ensure support coverage is detailed in Appendix~\ref{sec:appendix_initialization}.
In practice, we bypass explicit reward modeling and directly optimize the policy.
Following \citet{nicks2023language}, we define the implicit log-ratio $h_\theta(\mathbf{y}|\mathbf{x}) = \beta \log \frac{\pi_\theta(\mathbf{y}|\mathbf{x})}{\pi_{\text{ref}}(\mathbf{y}|\mathbf{x})}$, yielding the adversarial loss:
\begin{equation}
\label{eq:dpo_loss}
\mathcal{L}_{\text{DPO}} = -\mathbb{E}_{\mathcal{D}_{\text{DPO}}} \left[ \log \sigma \left( h_\theta(\mathbf{y}_w|\mathbf{x}) - h_\theta(\mathbf{y}_l|\mathbf{x}) \right) \right].
\end{equation}

% Minimizing this loss functions as implicit adversarial training, utilizing the detector's negative feedback to push the generation distribution from the AI region toward the human region.
Minimizing this loss functions as implicit adversarial training, pushing the generation distribution from the AI region toward the human region based on the detector's judgments.

% \noindent \textbf{Stage 4: Inference-Time Adversarial Refinement.}
% DPO-aligned text $\hat{T}$ may occasionally suffer from grammatical imperfections or minor semantic deviations, we introduce a refinement module that enhances quality while strictly preserving the detector's human-written prediction.

% \textbf{Candidate Generation.}
% We decompose $\hat{T}$ into sentences $S = \{s_1, \dots, s_n\}$. To ensure semantic fidelity, we employ an LLM polisher $\mathcal{G}$ to generate a grammatically optimized candidate $c_i$ for each $s_i$, anchored to the original machine-generated content $\mathcal{C}_{ref}$ given a polishing instruction $I$:
% \begin{equation}
% c_i = \mathcal{G}(s_i \mid \mathcal{C}_{ref}, I).
% \end{equation}

% \textbf{PPL-based Ranking Adversarial Refinement.} To optimize query efficiency, we prioritize refining sentences with the highest perplexity. Using a proxy model $\mathcal{M}_{ppl}$, we rank sentences in descending order. We employ a greedy verification strategy. For each sentence $s_i$, the candidate $c_i$ is accepted only if the detector $D$ retains high confidence in the human-written class for the updated text $\mathbf{T}'$:
% \begin{equation}
% S_{curr}[i] \leftarrow 
% \begin{cases} 
% c_i, & \text{if } D(\mathbf{T}') = \text{Human} \\
% s_i, & \text{otherwise}
% \end{cases}
% \end{equation}
% This enhances fluency and consistency without compromising evasion success.

\noindent \textbf{Stage 4: Inference-Time Adversarial Refinement.}
We introduce a refinement module that enhances text quality while strictly preserving the detector's human-written prediction.

\textbf{Candidate Generation.}
We decompose the DPO-aligned text $\hat{T}$ into a sentence sequence $S = \{s_1, \dots, s_n\}$. For each sentence $s_i$, an LLM polisher $\mathcal{G}$ generates a candidate $c_i$ conditioned on the original machine-generated content $\mathcal{C}_{ref}$ and a polishing instruction $I$, formulated as $c_i = \mathcal{G}(s_i \mid \mathcal{C}_{ref}, I)$.

\textbf{PPL-based Ranking Adversarial Refinement.}
To optimize query efficiency, we rank sentences by perplexity in descending order, prioritizing low-fluency sentences.
We apply a greedy strategy where $s_i$ is replaced by $c_i$ only if the detector $D$ maintains a human prediction. This enhances fluency without compromising attack success.

\section{Experiments}

\subsection{Experimental Setup}

\textbf{Datasets and Detectors.} To ensure a comprehensive evaluation \citep{zheng2025th}, we evaluate MASH on six domains from two benchmarks: MGTBench (Essay, Reuters, WP) \citep{he2024mgtbench} and MGT-Academic (STEM, Social, Humanity) \citep{liu2025generalization}. 
We test against three open-source detectors (RoBERTa, Binoculars, SCRN) and two commercial APIs (Writer, Scribbr).

\noindent \textbf{Baselines and Metrics.}
% We benchmark against 11 state-of-the-art methods, comprising five perturbation-based, one prompt-based, and five paraphrase-based attacks.
We benchmark against 11 state-of-the-art methods: five perturbation-based attacks (DeepWordBug~\citep{gao2018black}, TextBugger~\citep{li2018textbugger}, TextFooler~\citep{jin2020bert}, Charmer~\citep{Abad2024Charmer}, HMGC~\citep{Zhou2024_COLING}), one prompt-based attack (PromptAttack~\citep{xu2024an}), and five paraphrase-based attacks (DIPPER~\citep{krishna2023paraphrasing}, DPO-Evader~\citep{nicks2023language}, Toblend~\citep{huang2024toblend}, CoPA~\citep{fang-etal-2025-language}, GradEscape~\citep{meng2025gradescape}).

Following \citet{wang2025humanizing, meng2025gradescape}, we evaluate evasion performance via Attack Success Rate (ASR) and text quality via Perplexity (PPL), BERTScore \citep{zhang2019bertscore}, and GRUEN \citep{zhu2020gruen}.

\noindent \textbf{Implementation Details.}
MASH is initialized with BART-base \citep{lewis2020bart} and optimized via AdamW on a single NVIDIA RTX 3090 GPU.
% It takes 7.5 hours to train the paraphraser to attack a single target detector.
Detailed experimental configurations are provided in Appendix~\ref{app:exp_details}.

\begin{table*}[t]
  \centering

  \resizebox{\textwidth}{!}{%
  \begin{tabular}{l|cccc|cccc|cccc}
    \toprule
    \multirow{2}{*}{\textbf{Method}} & \multicolumn{4}{c|}{\textbf{MGTBench-Essay}} & \multicolumn{4}{c|}{\textbf{MGTBench-Reuters}} & \multicolumn{4}{c}{\textbf{MGTBench-WP}} \\
    \cmidrule(lr){2-5} \cmidrule(lr){6-9} \cmidrule(lr){10-13}
     & \textbf{ASR} $\uparrow$ & \textbf{PPL} $\downarrow$ & \textbf{GRUEN} $\uparrow$ & \textbf{BS} $\uparrow$ & \textbf{ASR} $\uparrow$ & \textbf{PPL} $\downarrow$ & \textbf{GRUEN} $\uparrow$ & \textbf{BS} $\uparrow$ & \textbf{ASR} $\uparrow$ & \textbf{PPL} $\downarrow$ & \textbf{GRUEN} $\uparrow$ & \textbf{BS} $\uparrow$ \\
    \midrule
    DeepWordBug & 0.13 & 63.02 & 0.4703 & 0.9339 & 0.02 & 88.83 & 0.5240 & 0.9325 & 0.51 & 32.85 & 0.5238 & 0.9497 \\
    TextBugger & 0.47 & 40.67 & 0.4572 & 0.9188 & 0.17 & 84.02 & 0.4043 & 0.8929 & 0.77 & 34.24 & 0.4936 & 0.9318 \\
    TextFooler & 0.38 & 77.18 & 0.4465 & 0.9399 & 0.29 & 181.46 & 0.3018 & 0.8700 & 0.59 & 50.20 & 0.4249 & 0.9379 \\
    DIPPER & 0.07 & 44.78 & 0.6481 & 0.9034 & 0.00 & 11.07 & 0.7180 & 0.9096 & 0.04 & 14.11 & 0.5773 & 0.9006 \\
    Toblend & 0.00 & 7.03 & 0.3777 & 0.7525 & 0.00 & 6.88 & 0.5581 & 0.7242 & 0.19 & 6.59 & 0.4814 & 0.7312 \\
    PromptAttack & 0.00 & 11.21 & 0.7776 & 0.9155 & 0.00 & 10.55 & 0.7909 & 0.9269 & 0.00 & 10.55 & 0.7256 & 0.9246 \\
    Charmer & 0.45 & 38.00 & 0.5532 & 0.9664 & 0.05 & 114.57 & 0.5828 & 0.9624 & 0.62 & 17.26 & 0.5785 & 0.9759 \\ 
    DPO-Evader & 0.00 & 5.14 & 0.1539 & 0.9343 & 0.00 & 4.53 & 0.1492 & 0.9257 & 0.00 & 5.41 & 0.2037 & 0.9230 \\
    HMGC & 0.04 & 57.80 & 0.5883 & 0.9226 & 0.01 & 82.60 & 0.5469 & 0.9006 & 0.15 & 41.82 & 0.5763 & 0.9319 \\
    GradEscape & 0.22 & 17.72 & 0.1345 & 0.9038 & 0.02 & 74.96 & 0.4811 & 0.8141 & 0.00 & 13.49 & 0.6516 & 0.9865 \\
    CoPA & 0.01 & 29.42 & 0.6891 & 0.8750 & 0.00 & 41.64 & 0.6942 & 0.8769 & 0.17 & 25.47 & 0.6325 & 0.8807 \\
    \rowcolor{gray!20} \textbf{Ours} & \textbf{0.95} & 18.90 & 0.6614 & 0.9004 & \textbf{0.73} & 9.09 & 0.6790 & 0.9015 & \textbf{0.90} & 20.80 & 0.6471 & 0.8974 \\
    \midrule
    \multirow{2}{*}{\textbf{Method}} & \multicolumn{4}{c|}{\textbf{MGT-Academic-Humanity}} & \multicolumn{4}{c|}{\textbf{MGT-Academic-Social Science}} & \multicolumn{4}{c}{\textbf{MGT-Academic-STEM}} \\
    \cmidrule(lr){2-5} \cmidrule(lr){6-9} \cmidrule(lr){10-13}
     & \textbf{ASR} $\uparrow$ & \textbf{PPL} $\downarrow$ & \textbf{GRUEN} $\uparrow$ & \textbf{BS} $\uparrow$ & \textbf{ASR} $\uparrow$ & \textbf{PPL} $\downarrow$ & \textbf{GRUEN} $\uparrow$ & \textbf{BS} $\uparrow$ & \textbf{ASR} $\uparrow$ & \textbf{PPL} $\downarrow$ & \textbf{GRUEN} $\uparrow$ & \textbf{BS} $\uparrow$ \\
    \midrule
    DeepWordBug & 0.07 & 30.09 & 0.4867 & 0.9520 & 0.11 & 32.59 & 0.5411 & 0.9514 & 0.07 & 22.64 & 0.4662 & 0.9543 \\
    TextBugger & 0.15 & 40.51 & 0.4031 & 0.9319 & 0.11 & 56.87 & 0.4425 & 0.9364 & 0.06 & 39.93 & 0.4041 & 0.9333 \\
    TextFooler & 0.11 & 72.33 & 0.3495 & 0.9289 & 0.10 & 136.89 & 0.3315 & 0.9121 & 0.07 & 81.24 & 0.3189 & 0.9255 \\
    DIPPER & 0.57 & 14.67 & 0.5782 & 0.9031 & 0.51 & 13.88 & 0.6833 & 0.9061 & 0.55 & 12.47 & 0.6039 & 0.8789 \\
    PromptAttack & 0.00 & 12.46 & 0.6893 & 0.9277 & 0.00 & 11.61 & 0.6904 & 0.9295 & 0.01 & 10.39 & 0.6529 & 0.9340 \\
    Charmer & 0.73 & 14.54 & 0.5504 & 0.9711 & 0.84 & 12.88 & 0.6409 & 0.9736 & 0.29 & 11.62 & 0.5437 & 0.9874 \\ 
    DPO-Evader & 0.06 & 9.82 & 0.4409 & 0.9206 & 0.31 & 6.80 & 0.3529 & 0.9271 & 0.14 & 8.37 & 0.5404 & 0.9329 \\
    HMGC & 0.03 & 48.02 & 0.5181 & 0.9155 & 0.07 & 58.85 & 0.5705 & 0.9170 & 0.03 & 42.69 & 0.5283 & 0.9257 \\
    GradEscape & 0.37 & 21.51 & 0.6151 & 0.9584 & 0.52 & 14.69 & 0.7015 & 0.9666 & 0.38 & 20.40 & 0.6162 & 0.9499 \\
    CoPA & 0.20 & 27.52 & 0.6388 & 0.8727 & 0.19 & 29.95 & 0.6346 & 0.8740 & 0.16 & 23.75 & 0.6446 & 0.8698 \\
    \rowcolor{gray!20} \textbf{Ours} & \textbf{0.87} & 20.67 & 0.6396 & 0.9048 & \textbf{0.98} & 17.54 & 0.7294 & 0.8993 & \textbf{1.00} & 20.08 & 0.6431 & 0.8194 \\
    \bottomrule
  \end{tabular}%
  }
  \caption{Evasion performance comparison against the fine-tuned RoBERTa detector (detection threshold $\tau = 0.5$).}
  \label{tab:main_results}
\end{table*}

\begin{table*}[t]
  \centering

  \resizebox{\textwidth}{!}{%
  \begin{tabular}{l|cccc|cccc|cccc}
    \toprule
    \multirow{2}{*}{\textbf{Method}} & \multicolumn{4}{c|}{\textbf{MGTBench-Essay}} & \multicolumn{4}{c|}{\textbf{MGTBench-Reuters}} & \multicolumn{4}{c}{\textbf{MGTBench-WP}} \\
    \cmidrule(lr){2-5} \cmidrule(lr){6-9} \cmidrule(lr){10-13}
     & \textbf{ASR} $\uparrow$ & \textbf{PPL} $\downarrow$ & \textbf{GRUEN} $\uparrow$ & \textbf{BS} $\uparrow$ & \textbf{ASR} $\uparrow$ & \textbf{PPL} $\downarrow$ & \textbf{GRUEN} $\uparrow$ & \textbf{BS} $\uparrow$ & \textbf{ASR} $\uparrow$ & \textbf{PPL} $\downarrow$ & \textbf{GRUEN} $\uparrow$ & \textbf{BS} $\uparrow$ \\
    \midrule
    DeepWordBug & 0.00 & 11.46 & 0.5472 & 0.9767 & 0.00 & 11.42 & 0.6739 & 0.9796 & 0.00 & 12.40 & 0.6179 & 0.9823 \\
    TextBugger & 0.00 & 12.61 & 0.5362 & 0.9644 & 0.00 & 13.65 & 0.6360 & 0.9578 & 0.00 & 14.43 & 0.6109 & 0.9706 \\
    TextFooler & 0.00 & 12.86 & 0.5743 & 0.9781 & 0.00 & 13.64 & 0.6710 & 0.9746 & 0.00 & 15.09 & 0.6027 & 0.9760 \\
    DIPPER & 0.23 & 9.97 & 0.6372 & 0.9041 & 0.16 & 11.08 & 0.7442 & 0.9006 & 0.33 & 12.63 & 0.6679 & 0.9086 \\
    Toblend & 0.04 & 16.44 & 0.3931 & 0.7583 & 0.21 & 13.34 & 0.5293 & 0.7302 & 0.22 & 13.22 & 0.5303 & 0.7372 \\
    PromptAttack & 0.53 & 10.00 & 0.7567 & 0.9169 & 0.55 & 9.98 & 0.8083 & 0.9327 & 0.35 & 10.30 & 0.6881 & 0.9264 \\
    Charmer & 0.50 & 11.93 & 0.5505 & 0.9846 & 0.63 & 12.29 & 0.6677 & 0.9820 & 0.52 & 14.95 & 0.6233 & 0.9838 \\
    DPO-Evader & 0.03 & 4.20 & 0.1305 & 0.9248 & 0.03 & 3.81 & 0.1641 & 0.9233 & 0.02 & 4.59 & 0.1591 & 0.9198 \\
    HMGC & 0.32 & 13.67 & 0.5545 & 0.9781 & 0.41 & 13.00 & 0.6428 & 0.9779 & 0.31 & 14.26 & 0.6208 & 0.9773 \\
    GradEscape & 0.64 & 14.45 & 0.6446 & 0.9558 & 0.62 & 13.27 & 0.7200 & 0.9594 & 0.37 & 12.29 & 0.6745 & 0.9683 \\
    CoPA & 0.11 & 8.95 & 0.2516 & 0.9165 & 0.05 & 9.24 & 0.4468 & 0.9164 & 0.16 & 10.08 & 0.3921 & 0.9156 \\
    \rowcolor{gray!20} \textbf{Ours} & \textbf{0.94} & 11.39 & 0.6755 & 0.9077 & \textbf{0.95} & 12.91 & 0.7122 & 0.9037 & \textbf{0.85} & 12.05 & 0.6883 & 0.9286 \\
    \midrule
    \multirow{2}{*}{\textbf{Method}} & \multicolumn{4}{c|}{\textbf{MGT-Academic-Humanity}} & \multicolumn{4}{c|}{\textbf{MGT-Academic-Social Science}} & \multicolumn{4}{c}{\textbf{MGT-Academic-STEM}} \\
    \cmidrule(lr){2-5} \cmidrule(lr){6-9} \cmidrule(lr){10-13}
     & \textbf{ASR} $\uparrow$ & \textbf{PPL} $\downarrow$ & \textbf{GRUEN} $\uparrow$ & \textbf{BS} $\uparrow$ & \textbf{ASR} $\uparrow$ & \textbf{PPL} $\downarrow$ & \textbf{GRUEN} $\uparrow$ & \textbf{BS} $\uparrow$ & \textbf{ASR} $\uparrow$ & \textbf{PPL} $\downarrow$ & \textbf{GRUEN} $\uparrow$ & \textbf{BS} $\uparrow$ \\
    \midrule
    DeepWordBug & 0.00 & 10.73 & 0.6387 & 0.9822 & 0.00 & 10.85 & 0.6050 & 0.9828 & 0.00 & 8.43 & 0.5746 & 0.9871 \\
    TextBugger & 0.00 & 11.81 & 0.6173 & 0.9675 & 0.00 & 13.39 & 0.5854 & 0.9628 & 0.00 & 10.42 & 0.5485 & 0.9680 \\
    TextFooler & 0.00 & 11.37 & 0.6369 & 0.9813 & 0.00 & 13.38 & 0.6143 & 0.9794 & 0.00 & 9.92 & 0.5776 & 0.9840 \\
    DIPPER & 0.53 & 11.93 & 0.6378 & 0.8917 & 0.38 & 11.63 & 0.6407 & 0.8891 & 0.49 & 10.86 & 0.6429 & 0.9048 \\
    PromptAttack & 0.83 & 11.13 & 0.7382 & 0.9270 & 0.86 & 10.67 & 0.7320 & 0.9276 & 0.81 & 9.33 & 0.6904 & 0.9318 \\
    Charmer & 0.71 & 12.91 & 0.6304 & 0.9831 & 0.77 & 12.37 & 0.5960 & 0.9832 & 0.71 & 10.16 & 0.5592 & 0.9841 \\
    DPO-Evader & 0.08 & 5.79 & 0.2598 & 0.9502 & 0.06 & 5.17 & 0.1590 & 0.9440 & 0.04 & 5.02 & 0.1795 & 0.9525 \\
    HMGC & 0.69 & 15.05 & 0.6024 & 0.9732 & 0.68 & 19.14 & 0.5847 & 0.9715 & 0.73 & 15.25 & 0.5384 & 0.9708 \\
    GradEscape & 0.87 & 17.58 & 0.7000 & 0.9416 & 0.88 & 16.86 & 0.6649 & 0.9422 & 0.95 & 20.92 & 0.6546 & 0.9286 \\
    CoPA & 0.25 & 9.25 & 0.3889 & 0.9273 & 0.17 & 9.52 & 0.3438 & 0.9258 & 0.23 & 7.85 & 0.3377 & 0.9328 \\
    \rowcolor{gray!20} \textbf{Ours} & \textbf{0.95} & 13.47 & 0.7025 & 0.9074 & \textbf{0.95} & 13.31 & 0.7014 & 0.9073 & \textbf{0.99} & 14.02 & 0.6434 & 0.8864 \\
    \bottomrule
  \end{tabular}%
  }
    \caption{Evasion performance comparison against the Binoculars detector.}
      \label{tab:main_results_bin}
\end{table*}

\begin{table}[t]
\centering
\small

\resizebox{\columnwidth}{!}{
\begin{tabular}{llcccccc}
\toprule
\multirow{2}{*}{\textbf{Detector}} & \multirow{2}{*}{\textbf{Dataset}} & \multirow{2}{*}{\textbf{ASR}} & \multirow{2}{*}{\textbf{BS}} & \multicolumn{2}{c}{\textbf{No Attack}} & \multicolumn{2}{c}{\textbf{With Attack}} \\
\cmidrule(lr){5-6} \cmidrule(lr){7-8}
& & & & \textbf{PPL} & \textbf{GRUEN} & \textbf{PPL} & \textbf{GRUEN} \\
\midrule
\multirow{6}{*}{Scribbr} 
& Essay    & 0.95 & 0.8800 & 7.33 & 0.5903 & 16.89 & 0.6599 \\
& Reuters  & 0.59 & 0.8740 & 8.87 & 0.6964 & 8.54  & 0.7318 \\
& WP       & 0.90 & 0.8844 & 8.18 & 0.7133 & 17.18 & 0.6302 \\
& Humanity & 0.94 & 0.9288 & 10.82& 0.5899 & 19.61 & 0.6712 \\
& Social   & 0.98 & 0.8730 & 9.50 & 0.6714 & 18.53 & 0.7333 \\
& STEM     & 1.00 & 0.7840 & 8.61 & 0.5690 & 20.29 & 0.6442 \\
\midrule
\multirow{6}{*}{Writer} 
& Essay    & 0.90 & 0.8722 & 7.41 & 0.5727 & 15.87 & 0.6806 \\
& Reuters  & 0.66 & 0.8624 & 7.91 & 0.6440 & 9.16  & 0.7440 \\
& WP       & 0.91 & 0.8947 & 8.79 & 0.6524 & 15.78 & 0.6944 \\
& Humanity & 0.86 & 0.9067 & 8.39 & 0.6511 & 20.32 & 0.6246 \\
& Social   & 0.96 & 0.9062 & 7.60 & 0.5983 & 16.42 & 0.7092 \\
& STEM     & 0.99 & 0.7832 & 7.91 & 0.6102 & 18.29 & 0.6741 \\
\bottomrule
\end{tabular}
}

\caption{Commercial detector evasion.}
\label{tab:commercial_attack_results}
\end{table}

\subsection{Performance Evaluation}

\textbf{Attack Effectiveness.}
% We evaluate MASH's evasion performance across a range of detection paradigms. 
As shown in Table~\ref{tab:main_results}, MASH outperforms the strongest baselines by an average of 29.7\% against the fine-tuned RoBERTa detector, while maintaining superior text quality.
The limitations of perturbation-based attacks are illustrated in Figure~\ref{fig:short_and_long} of Appendix~\ref{sec:appendix_additional_results}. Paraphrase-based methods also perform poorly.
% Specifically, methods like DIPPER struggle with out-of-distribution (OOD) generalization, and CoPA's indirect proxy-guidance fails to capture target-specific decision boundaries
Among existing paraphrase-based baselines, DIPPER struggles with out-of-distribution (OOD) generalization, while CoPA's indirect proxy-guidance fails to capture target-specific decision boundaries.

We further evaluate MASH against zero-shot detectors. As detailed in Table~\ref{tab:main_results_bin}, MASH surpasses the strongest baselines by an average of 19.0\%.
This suggests that transferring to human style naturally corrects the statistical patterns that zero-shot detectors rely on.
% Notably, MASH achieves near-perfect evasion rates in structured domains (e.g., 0.99 on STEM).

% Given the high query costs of commercial APIs, we first benchmarked baseline performance using a balanced subset of 50 human-written and 50 AI-generated samples per domain (see Table~\ref{tab:commercial_performance} in Appendix~\ref{sec:appendix_additional_results}). Based on these findings, we targeted the two most robust domains, Essay and WP, for evasion attacks. As shown in Table~\ref{tab:commercial_attack_results}, MASH successfully compromises these high-performing detectors, achieving an ASR surpassing 90\% while maintaining high semantic fidelity.

For commercial detectors, we assess the detection capabilities of Writer and Scribbr, which exhibit high precision across all six domains (see Table~\ref{tab:commercial_performance} in Appendix~\ref{sec:appendix_additional_results}), confirming their robustness. As shown in Table~\ref{tab:commercial_attack_results}, MASH evades these detectors with an average ASR of 89\%, while preserving semantic consistency. These results confirm that MASH generalizes effectively to real-world commercial systems under black-box settings.

\begin{figure}[t]
\includegraphics[width=\columnwidth]{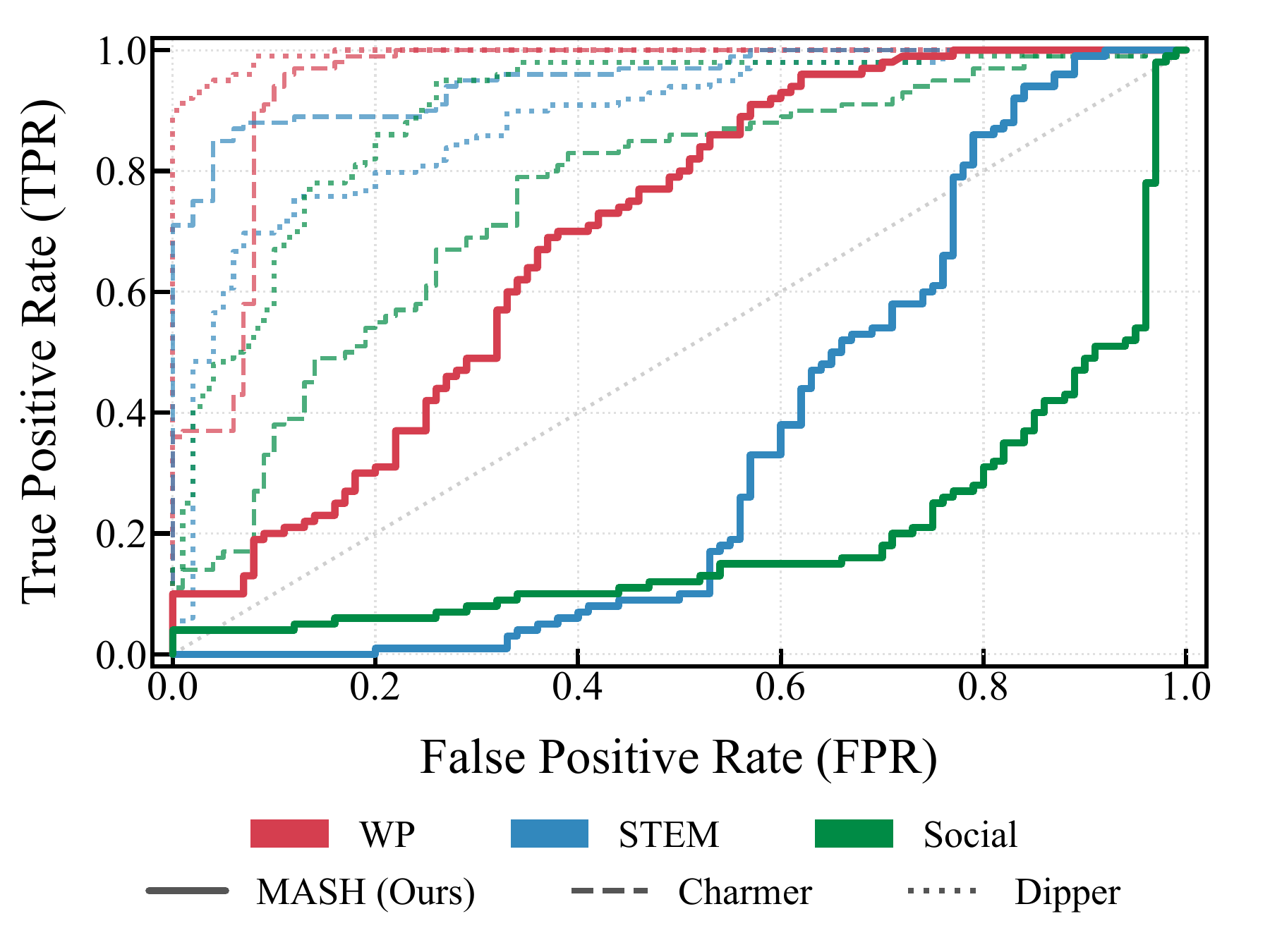}
\caption{ROC curves comparing the evasion effectiveness of MASH against baseline attacks on the RoBERTa detector across the WP, STEM, and Social domains. Note that a lower TPR corresponds to a higher ASR.}
\label{fig:roc}
\end{figure}

\noindent \textbf{Interpretation of AUROC.}
In real-world detection scenarios, the utility of a detector hinges on its capacity to maintain a high True Positive Rate (TPR) while rigorously suppressing the False Positive Rate (FPR) to prevent wrongful accusations against human authors. To evaluate this, we benchmarked our method against two highly efficient baselines, Charmer and DIPPER. Figure~\ref{fig:roc} illustrates the ROC curves under these three distinct attack modalities. Quantitatively, at a stringent 1\% FPR threshold, MASH suppresses the TPR to near-zero levels, particularly in the Social and STEM domains. This confirms that MASH effectively neutralizes the detector's discriminative capability, even under conservative operational constraints.

% In real-world detection scenarios, the utility of a detector hinges on its capacity to maintain a high True Positive Rate (TPR) while rigorously suppressing the False Positive Rate (FPR) to prevent wrongful accusations against human authors. To evaluate this trade-off, we benchmark MASH against two strong baselines, Charmer and DIPPER. Figure~\ref{fig:roc} illustrates the ROC curves under these three attack methods. At a stringent 1\% FPR threshold, MASH suppresses the TPR to near-zero levels, particularly in the Social and STEM domains. This confirms that MASH effectively neutralizes the detector's discriminative capability, even under conservative operational constraints.

\begin{table}[t]
\centering
\small

\resizebox{\columnwidth}{!}{
\begin{tabular}{lccc}
\toprule
\textbf{Method} & \textbf{GPU Memory} & \textbf{Inference Time} & \textbf{Query Cost} \\
\midrule
HMGC & 41046 & 86 & 0 \\
Charmer & 30094 & 189.9 & 18827 \\
DPO\_Evader & 22699 & 0.2 & 0 \\
DIPPER & 22241 & 58.6 & 0 \\
CoPA & 15425 & 6.7 & 0 \\
ToBlend & 9223 & 31.8 & 0 \\
TextBugger & 3753 & 78.3 & 5275.7 \\
TextFooler & 3753 & 40.6 & 2913.8 \\
GradEscape & 3133 & 1.7 & 0 \\
DeepWordBug & 1697 & 29.1 & 1018.4 \\
\midrule
\rowcolor{gray!20} \textbf{MASH } & 3139 & 1.7 & 0 \\
\bottomrule
\end{tabular}
}

\caption{Comparison of resource usage. Units: GPU Memory (MiB), Inference Time (seconds/sample), and Query Cost (queries/sample).}
\label{tab:method_comparison_sorted}

\end{table}

\noindent \textbf{Efficiency Analysis.}
% To evaluate the practical feasibility of our framework, we benchmark the inference overhead of MASH against both perturbation-based and paraphrase-based baselines. 
% We evaluate MASH excluding the Stage 4 refinement module.
% As reported in Table~\ref{tab:method_comparison_sorted}, perturbation-based attacks (e.g., Charmer) incur prohibitive query costs and significant inference latency, while certain paraphrase-based methods (e.g., DIPPER) demand substantial GPU memory. 
% In contrast, MASH demonstrates superior efficiency, achieving a lightweight memory footprint ($\sim$3GB) and rapid inference speed (1.7s), comparable to GradEscape.
To evaluate practical feasibility, we exclude the Stage 4 refinement module (optional for users prioritizing ASR) and compare inference overhead. As reported in Table~\ref{tab:method_comparison_sorted}, perturbation-based attacks (e.g., Charmer) incur prohibitive query costs, while paraphrase-based methods (e.g., DIPPER) demand substantial GPU memory. In contrast, MASH achieves a lightweight memory footprint ($\sim$3GB) and rapid inference speed (1.7s).
For a comprehensive comparison including training and Stage 4 overhead, Figure~\ref{fig:query_efficiency} in Appendix~\ref{sec:appendix_additional_results} shows that MASH amortizes its overhead effectively: as attack samples scale to $10^4$, the average query cost per sample becomes negligible.

% \begin{figure}[t]
% \includegraphics[width=\columnwidth]{asset/图片17.pdf}
% \caption{t-SNE visualization of feature distributions from the RoBERTa detector in the Social domain.}
% \label{fig:signle_vis}
% \end{figure}

% \noindent \textbf{Visualization of Rewritten Texts.} The t-SNE visualization (see Figure~\ref{fig:signle_vis}) demonstrates a significant distribution shift: adversarial samples migrate from the distinct AI cluster to indistinguishably blend with the human manifold. 
% For additional results, please refer to Figure~\ref{fig:tsne} in Appendix~\ref{sec:appendix_additional_results}.
% This confirms that MASH effectively bridges the feature gap, rendering the detector's decision boundary ineffective.

\subsection{Ablation Study}
\label{sec:ablation}

\begin{figure}[t]
    \centering
    \includegraphics[width=\columnwidth]{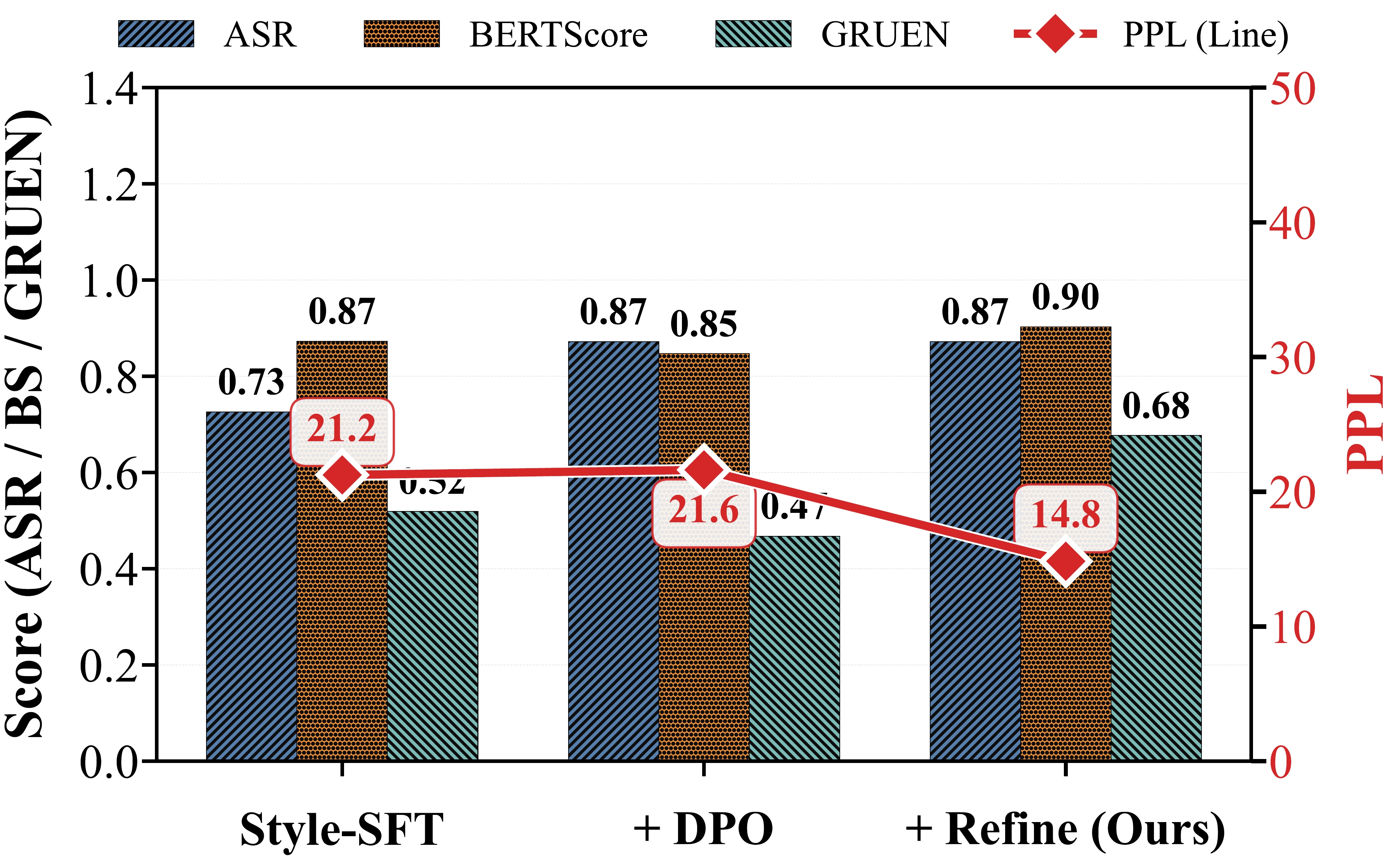} % 请替换为您的文件名
  \caption{Ablation study of MASH components, averaged across three detectors (RoBERTa, Binoculars, SCRN) and six datasets.}
  \label{fig:ave_ablation}
\end{figure}

% Figure~\ref{fig:ave_ablation} validates distinct component roles: while DPO serves as the primary driver for evasion, it often degrades fluency. Crucially, Inference-Time Refinement resolves this trade-off, significantly recovering linguistic quality while maintaining robust evasion success.

\textbf{Impact of MASH Components.}
% Figure~\ref{fig:ave_ablation} validates the distinct roles of the components: DPO significantly improves the paraphraser's ASR, while the refinement module effectively restores the text quality degraded by DPO. For detailed ablation results, please refer to Table~\ref{tab:ablation_study} in Appendix~\ref{sec:add_abla}.
% Figure~\ref{fig:ave_ablation} validates the contribution of each component: DPO significantly improves ASR, while the refinement module restores the text quality degraded by DPO. Detailed ablation results are provided in Table~\ref{tab:ablation_study} in Appendix~\ref{sec:add_abla}.
As shown in Figure~\ref{fig:ave_ablation}, DPO significantly improves ASR, while the refinement module restores the text quality degraded by DPO. Detailed ablation results are provided in Table~\ref{tab:ablation_study} (Appendix~\ref{sec:add_abla}).

\begin{figure}[t!] % 注意这里加了叹号，表示“强制”尝试
    \centering
    % 第一张图
    \includegraphics[width=0.95\columnwidth]{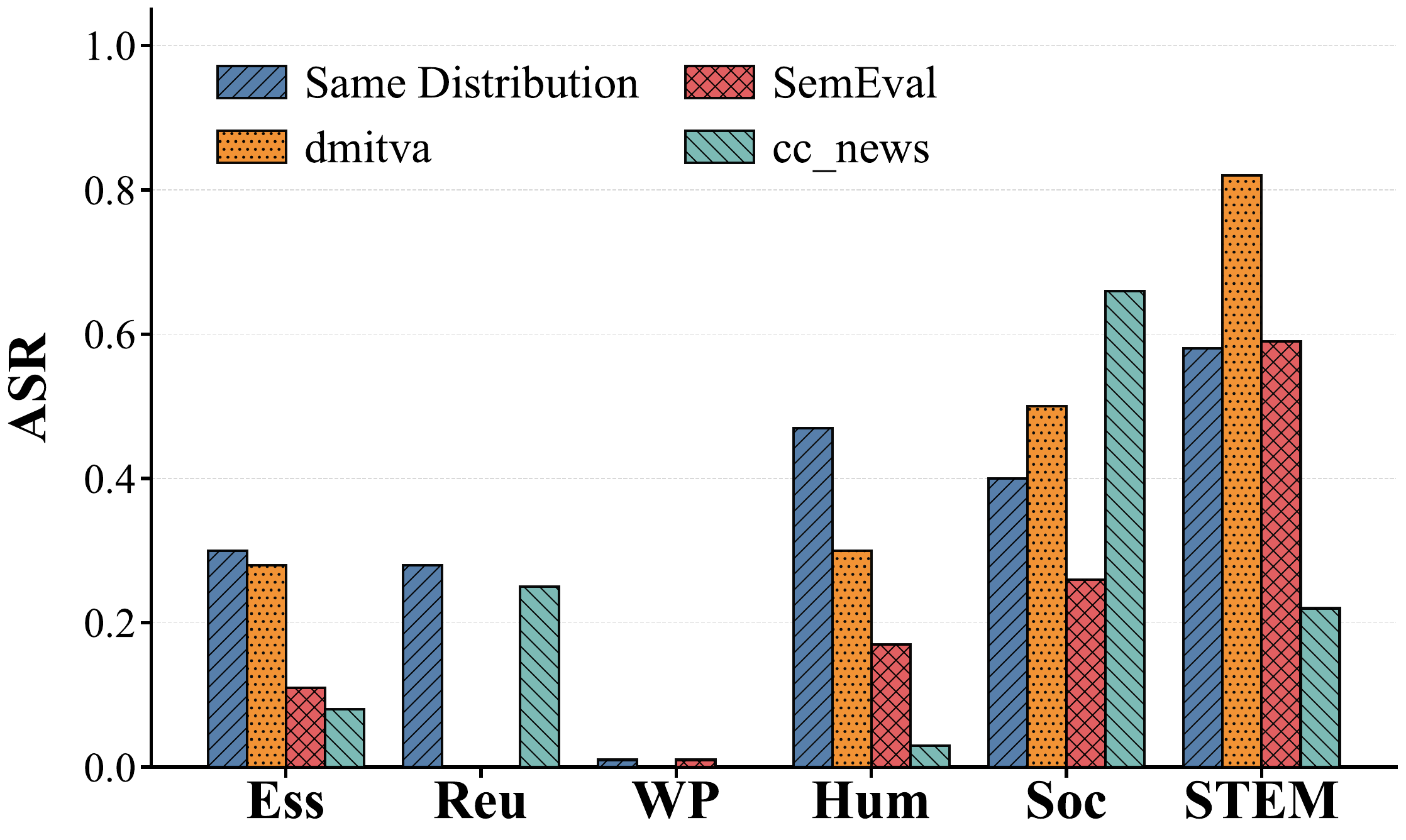} % 稍微缩小一点比例例如0.95能更容易塞进去
    % \vspace{-5pt} % 减少图和标题的距离
    \caption{Impact of data source selection.}
    \label{fig:differt_data_source}
    
    \vspace{10pt} % 两张图之间的距离
    
    % 第二张图
    \includegraphics[width=0.95\columnwidth]{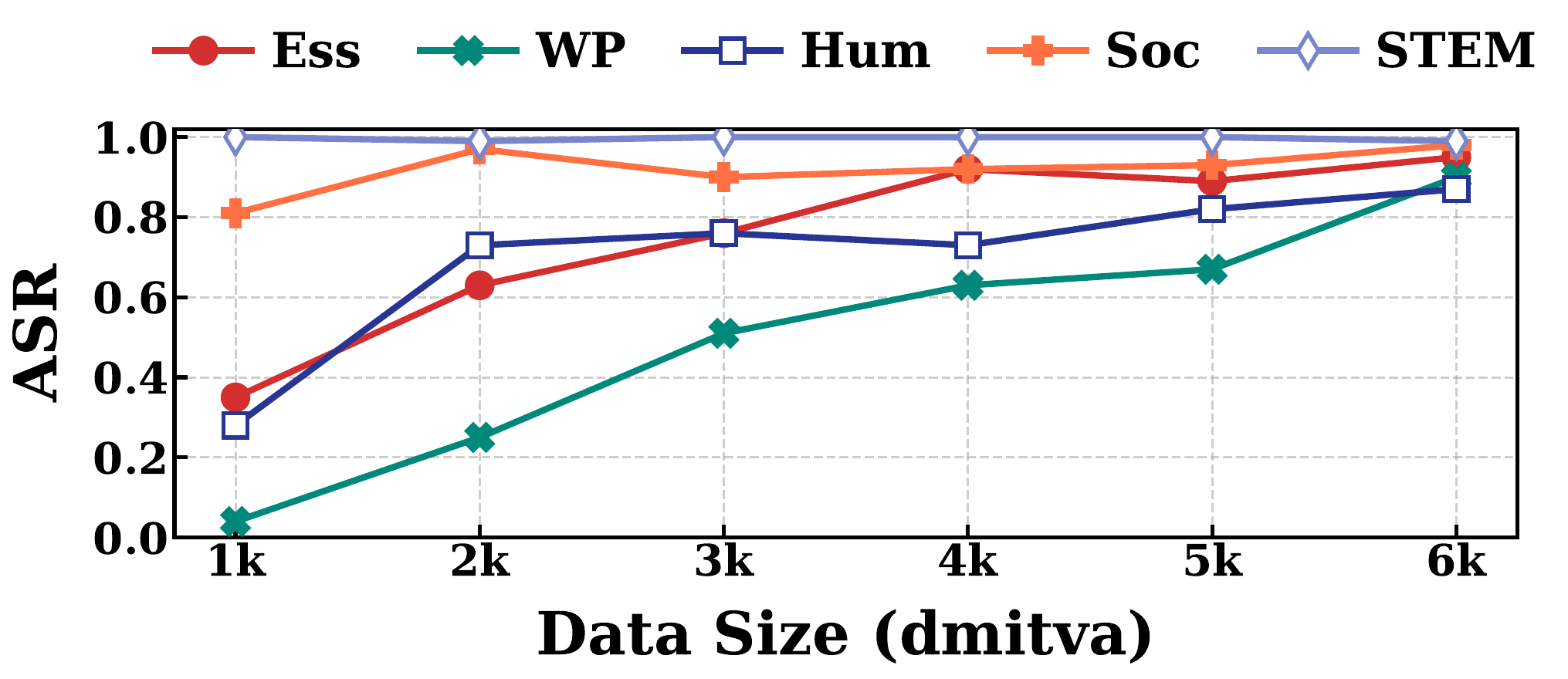}
    % \vspace{-5pt} % 减少图和标题的距离
    \caption{Impact of training data size.}
    \label{fig:different_data_scale}
    
    \vspace{-10pt} % 减少整个Figure环境与下方正文的距离
\end{figure}

% \begin{figure}[t]
%     \centering
%     \includegraphics[width=\columnwidth]{asset/图片10.pdf} % 请替换为您的文件名
%   \caption{Impact of data source selection.}
%   \label{fig:differt_data_source}
% \end{figure}

% \begin{figure}[t]
%     \centering
%     \includegraphics[width=\columnwidth]{asset/图片11.pdf} % 请替换为您的文件名
%     \caption{Impact of training data size.}
%     \label{fig:different_data_scale}
% \end{figure}

\noindent \textbf{Impact of Data Source and Size.} 
We constructed the training data from four sources (see Figure~\ref{fig:differt_data_source}). To establish a fair comparison, we downsampled open-source datasets to match the limited ``Same Distribution'' (i.e., matching the detector's training domain) baseline counts.
% Essay (362), Reuters (381), WP (292), Humanity (430), Social (220), and STEM (605). 
% Experimental results indicate that generic datasets can differ in effectiveness across domains.
Experimental results reveal that the effectiveness of generic datasets varies substantially across domains.
Notably, for the RoBERTa detector trained on Reuters, samples from 
\textit{dmitva}\footnote{\url{https://huggingface.co/datasets/dmitva/human_ai_generated_text}} and SemEval \citep{wang2024semeval} were predominantly classified as machine-generated, rendering them unsuitable for initializing style transfer in this domain.
In contrast, \textit{cc\_news}\footnote{\url{https://huggingface.co/datasets/vblagoje/cc_news}} proved effective. Figure~\ref{fig:different_data_scale} further confirms that performance scales with data size.
% where structured domains saturate significantly faster than stylistically complex tasks.

\begin{figure*}[t]
  \centering  % 图片居中
  \includegraphics[width=1\linewidth]{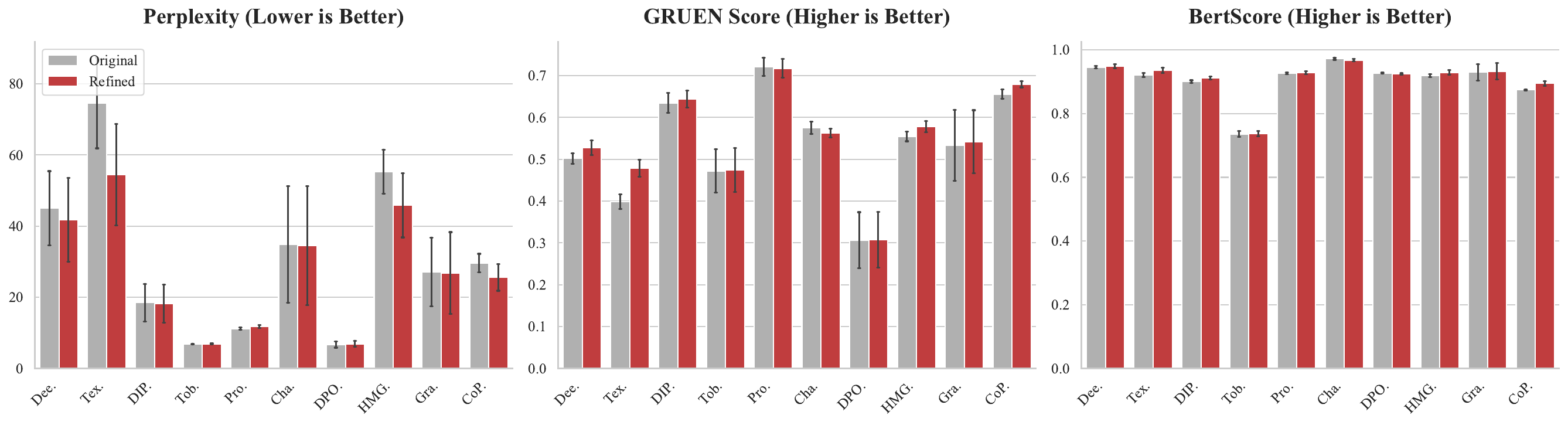} 
  \caption{Generalizability of inference-time refinement (Stage 4) to baseline attack methods, averaged across six datasets. Gray and red bars denote original and refined outputs, respectively; error bars indicate standard deviation.}
  \label{fig:stage4_baseline} %以此标签引用图片
\end{figure*}

\noindent \textbf{Generalizability of Refinement Module.}
To assess whether Stage 4 generalizes beyond MASH, we apply it to outputs generated by other baseline attack methods. As shown in Figure~\ref{fig:stage4_baseline}, the refinement module consistently improves text quality across most baselines.
% Perturbation-based methods benefit most significantly
% TextFooler achieves PPL reductions up to 65\% (136.89 to 47.22 on Social) and GRUEN gains exceeding 0.19.
% Methods already producing fluent text (e.g., PromptAttack, DIPPER) show marginal gains, indicating Stage 4 primarily addresses grammatical artifacts.
Detailed results are provided in Table \ref{tab:refinement_comparison} (Appendix \ref{sec:add_abla}).

% \noindent \textbf{Alternative Refinement Models.}
% % To reduce costs, we evaluate open-source alternatives to GPT-5 for Stage 4.
% % As shown in Table \ref{tab:roberta_quality_comparison} (Appendix \ref{sec:add_abla}), Qwen2.5-7B-Instruct achieves competitive GRUEN scores while Meta-Llama-3-8B-Instruct also improves over original outputs. However, GPT-5 consistently yields superior BERTScore. These results suggest open-source refiners are viable under cost constraints.
% To reduce costs, we evaluate open-source refinement alternatives. As shown in Table~\ref{tab:roberta_quality_comparison} (Appendix~\ref{sec:add_abla}), Qwen2.5-7B-Instruct achieves competitive GRUEN scores, and Meta-Llama-3-8B-Instruct improves over original outputs. However, GPT-5 yields superior BERTScores. These results confirm open-source refiners are viable under cost constraints.

\noindent \textbf{Alternative Refinement Models.}
To reduce the cost of Stage 4, we evaluate open-source alternatives. As shown in Table~\ref{tab:roberta_quality_comparison} (Appendix~\ref{sec:add_abla}), Qwen2.5-7B-Instruct improves GRUEN from 0.468 to 0.631 on average, approaching GPT-5's 0.667. However, GPT-5 maintains a clear advantage in BERTScore (0.887 vs. 0.857). 
These open-source alternatives remain viable under budget constraints.

\subsection{Universality and Transferability Analysis}
\label{sec:uni_and_trans}

% \begin{figure}[t!] % 注意这里加了叹号，表示“强制”尝试
%     \centering
%     % 第一张图
%     \includegraphics[width=0.95\columnwidth]{asset/图片6.pdf} % 稍微缩小一点比例例如0.95能更容易塞进去
%     % \vspace{-5pt} % 减少图和标题的距离
%     \caption{Cross-domain evasion transferability.}
%     \label{fig:domain_transfer}
    
%     \vspace{10pt} % 两张图之间的距离
    
%     % 第二张图
%     \includegraphics[width=0.95\columnwidth]{asset/图片5.pdf}
%     % \vspace{-5pt} % 减少图和标题的距离
%     \caption{Cross-detector evasion transferability.}
%     \label{fig:detector_transfer}
    
%     \vspace{-10pt} % 减少整个Figure环境与下方正文的距离
% \end{figure}

\begin{figure}[t]
    \centering
    \includegraphics[width=\columnwidth]{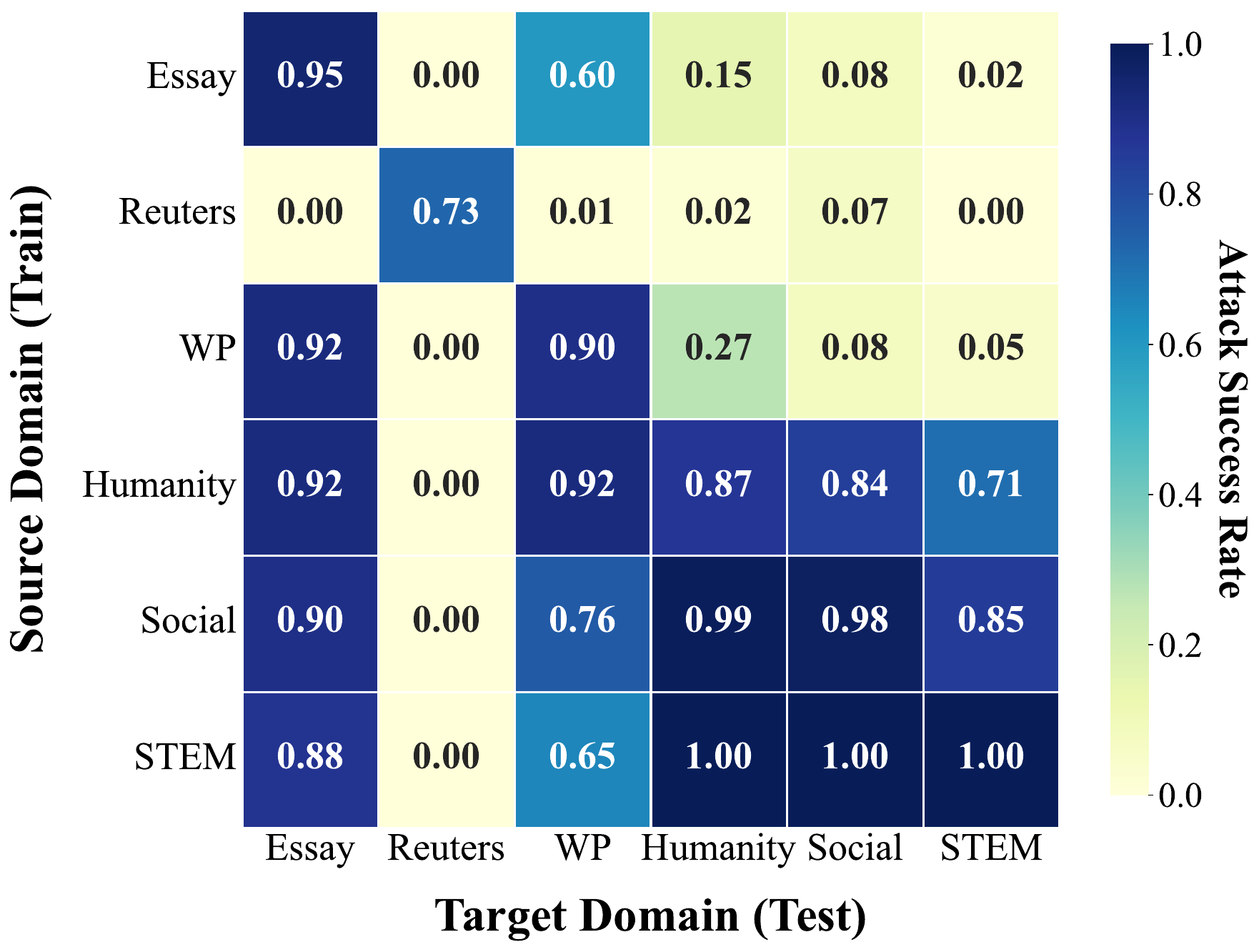} % 请替换为您的文件名
    \caption{Cross-domain evasion transferability.}
    \label{fig:domain_transfer}
\end{figure}

\noindent \textbf{Transferability across Domains.}
We evaluate cross-domain transferability by testing MASH against RoBERTa detectors fine-tuned on source domains and evaluated on different target domains.
As illustrated in Figure~\ref{fig:domain_transfer}, the results demonstrate robust generalization across most datasets. However, cross-domain scenarios involving Reuters exhibit an ASR approaching zero. We provide a detailed analysis of this phenomenon in Appendix~\ref{app:data_analysis}.

\begin{figure}[t]
    \centering
    \includegraphics[width=\columnwidth]{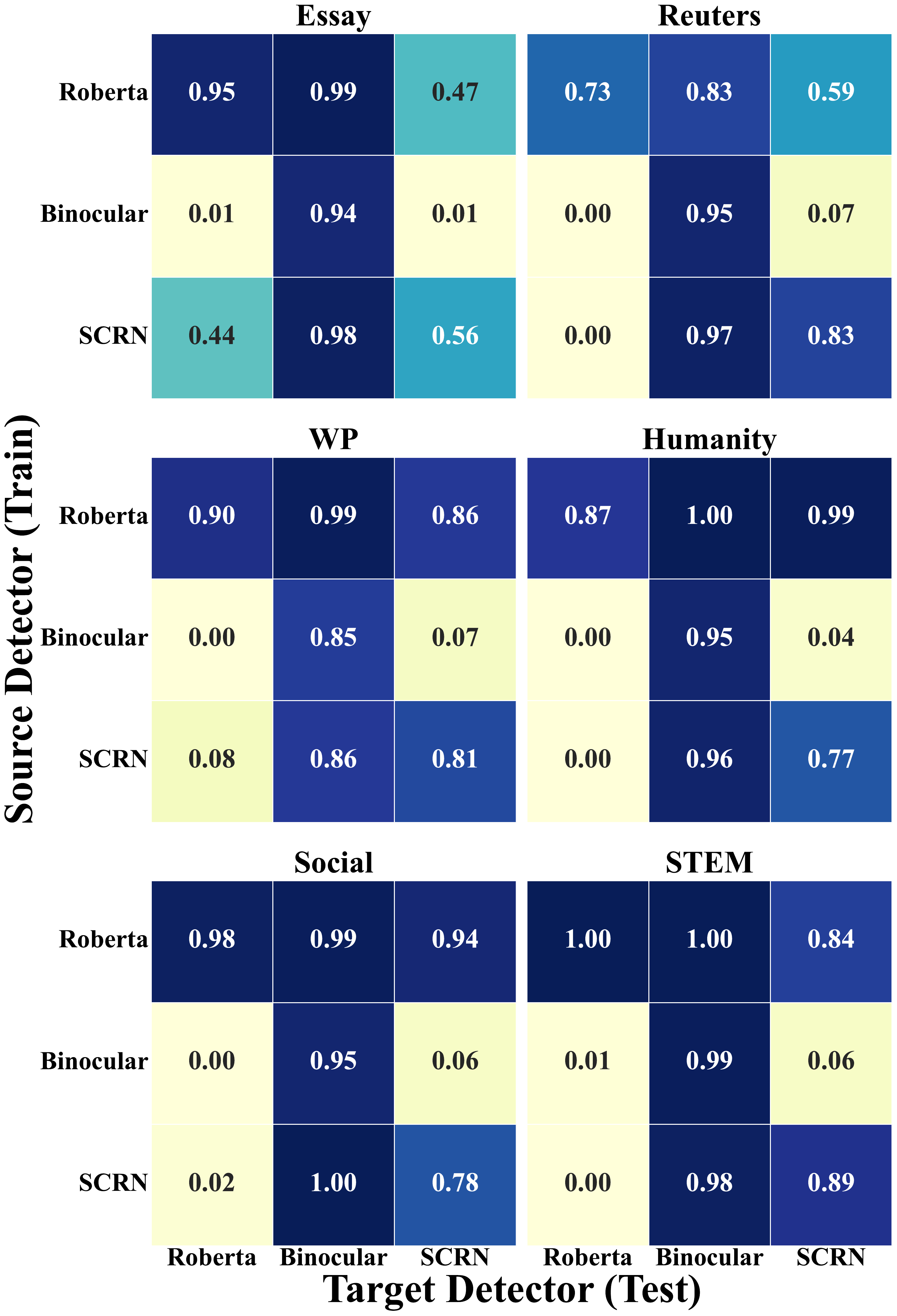} % 请替换为您的文件名
  \caption{Cross-detector evasion transferability.}
\label{fig:detector_transfer} \end{figure}

% \noindent \textbf{Transferability across Detectors.} To assess whether MASH captures universal human-like characteristics or merely overfits specific detector boundaries, we evaluated cross-detector transferability. Figure~\ref{fig:detector_transfer} reveals a distinct asymmetry: optimization against supervised detectors transfers to zero-shot methods by implicitly correcting statistical artifacts, whereas the reverse fails. This confirms that statistical alignment alone is insufficient to mimic complex human stylistic features.
\noindent \textbf{Transferability across Detectors.} To assess whether MASH captures universal human-like characteristics or merely overfits specific detector boundaries, we evaluate cross-detector transferability. We additionally introduce SCRN \citep{huang-etal-2024-ai}, another RoBERTa-based fine-tuned detector.
% , to provide a more comprehensive evaluation. 
As shown in Figure~\ref{fig:detector_transfer}, attacks optimized against supervised detectors (RoBERTa, SCRN) transfer effectively to zero-shot methods (Binoculars), whereas the reverse fails. 
% This asymmetry suggests that supervised detector optimization implicitly corrects statistical artifacts exploited by zero-shot methods, while statistical alignment alone is insufficient to mimic human stylistic features captured by supervised detectors.
This asymmetry suggests that supervised detector optimization implicitly corrects statistical artifacts exploited by zero-shot methods.
% while statistical alignment alone cannot capture human stylistic features.

\subsection{Defense via Adversarial Training}

To mitigate risks, we explore adversarial training by fine-tuning detectors on MASH-polished samples. This defense reduces ASR from 92\% to near-zero, but degrades accuracy on clean samples by 18.7\% on average. Details are provided in Appendix~\ref{sec:appendix_defense}.

\section{Conclusion}
We introduce MASH, a black-box framework that reformulates detector evasion as a style transfer task.
Experiments show MASH outperforms state-of-the-art baselines against both open-source and commercial detectors, exposing the fragility of current defenses.
MASH serves as both a red-teaming tool and a foundation for more robust detection.

\section*{Limitations}

\noindent \textbf{Limited scope of the experiments.} While we have conducted a comprehensive evaluation across five representative detectors, six diverse domains, and eleven state-of-the-art baseline methods, the landscape of AI-generated text detection is rapidly evolving. We acknowledge that our benchmarking represents a snapshot of the current security landscape. Future work could extend the MASH framework to test against emerging detection paradigms and novel adversarial attack strategies, further validating its generalization capabilities in dynamic adversarial environments.

\noindent \textbf{Prerequisites for Distribution Alignment.} The effectiveness of our Style-SFT module relies on the availability of a seed set of human-written texts that theoretically fall within the target detector's ``human'' decision boundary. As indicated in our ablation study (see Figure~\ref{fig:differt_data_source}), significant distribution shifts in generic open-source corpora can hinder initialization. However, this dependency is not a methodological flaw but a logical prerequisite: evasion is only mathematically meaningful against detectors that maintain a reasonable False Positive Rate (FPR). If a detector fails to correctly classify ground-truth human text (i.e., exhibits high FPR), the adversarial objective of ``mimicking human style'' becomes redundant.

\noindent \textbf{Linguistic Generalization.} Our current evaluations are concentrated on English-language benchmarks, including MGTBench and MGT-Academic. The applicability of our style-injection and DPO alignment mechanisms to languages with complex morphological structures or low-resource settings remains unexplored. Future work could investigate the cross-lingual transferability of MASH to assess its robustness in diverse linguistic contexts.

\section*{Ethics Statement}

This research investigates the vulnerabilities of current AI-generated text detectors under black-box settings. The primary motivation for proposing MASH is to serve as a red-teaming framework to expose the fragility of existing detection paradigms.
We acknowledge that the techniques described herein could potentially be misused for academic dishonesty or disinformation dissemination. However, we believe that security through obscurity is unsustainable; publicly disclosing these vulnerabilities is essential for the community to develop more robust detection systems.

Furthermore, all experiments in this work utilize open-source datasets and models, strictly adhering to their usage policies and data privacy regulations. We urge developers of detection systems to move beyond simple statistical features and incorporate adversarial training to defend against evasion via style transfer.


\bibliography{custom}

\appendix

\section{Theoretical Analysis}

This section serves as a theoretical supplement to the methodology presented in Section \ref{sec:mash_methodology}. We provide formal justifications for the MASH framework, specifically elucidating the mechanism of DPO-based evasion, the construction of preference pairs, and the necessity of the multi-stage pipeline.

\subsection{Distribution Alignment via DPO}
\label{sec:appendix_alignment}

To theoretically ground our approach, we bridge the gap between the adversarial evasion goal and the DPO optimization objective.
In our framework, the preference labels are derived from a black-box detector $D(\mathbf{y}) \in [0, 1]$, which estimates the probability of $\mathbf{y}$ being AI-generated. We posit that the underlying human preference follows a Bradley-Terry model driven by an implicit reward function $r(\mathbf{x}, \mathbf{y})$.
Following \citet{wang2025humanizing}, we formulate this reward as a scaled and shifted version of the detector's confidence score. For theoretical convenience and without loss of generality, we align the reward definition in Eq.~\ref{eq:reward_def} with the Bradley-Terry preference model by adopting the following equivalent form: 
\begin{equation} r(\mathbf{x}, \mathbf{y}) = C \cdot (1 - D(\mathbf{y})), \end{equation}
where $C$ is a sufficiently large constant.
Under this formulation, maximizing the implicit reward is mathematically equivalent to minimizing the detector's AI probability score. Crucially, the large constant $C$ renders the Bradley-Terry probability nearly deterministic, thereby theoretically aligning the DPO optimization objective with our adversarial evasion goal.

\begin{theorem}[Optimality of Evasion]
Let $\pi_{\text{ref}}$ be the reference policy and $\pi^*$ be the optimal policy minimizing the DPO loss. The generated distribution $\pi^*$ aligns with the human-like regions defined by the detector $D$ (i.e., regions where $D(\mathbf{y}) \to 0$).
\end{theorem}

\begin{proof}
As derived by \citet{rafailov2023direct}, the optimal policy $\pi^*$ for the KL-constrained reward maximization problem is uniquely determined by the reward function and the reference policy, taking the following closed-form solution:
\begin{equation}
    \pi^*(\mathbf{y}|\mathbf{x}) = \frac{1}{Z(\mathbf{x})} \pi_{\text{ref}}(\mathbf{y}|\mathbf{x}) \exp\left( \frac{1}{\beta} r(\mathbf{x}, \mathbf{y}) \right),
    \label{eq:dpo_optimal}
\end{equation}
where $Z(\mathbf{x})$ is the partition function. Substituting our evasion-oriented implicit reward $r \propto (1 - D(\mathbf{y}))$ into Eq.~\ref{eq:dpo_optimal}:
\begin{equation}
    \pi^*(\mathbf{y}|\mathbf{x}) \propto \pi_{\text{ref}}(\mathbf{y}|\mathbf{x}) \exp\left( \frac{C}{\beta} (1 - D(\mathbf{y})) \right).
\end{equation}

As the coefficient $\frac{C}{\beta}$ is large, the exponential term acts as a sharp filter. It amplifies the probability mass of samples where $1 - D(\mathbf{y}) \approx 1$ (i.e., $D(\mathbf{y}) \approx 0$, corresponding to human-written text) and suppresses regions where $D(\mathbf{y}) \approx 1$. Consequently, minimizing the DPO loss forces the generator's distribution to converge towards the human distribution characterized by the detector.
\end{proof}

\begin{table}[t]
\centering
\small

% 关键点：使用 resizebox 确保表格左右撑满栏宽，实现你要的“边界对齐”
\resizebox{\columnwidth}{!}{
\begin{tabular}{lccccc}
\toprule
 & \textbf{Essay} & \textbf{Reuters} & \textbf{WP} & \textbf{Humanity} & \textbf{Social} \\
\midrule
Stage 1 & 0.83 & 0.09 & 0.68 & 0.68 & 0.92 \\
% 使用缩写 AP 和 HN，避免文字过长导致 resizebox 缩放过度
Stage 2 (AP)$^{\dagger}$ & 0.87 & 0.61 & 0.84 & 0.66 & 0.87 \\
Stage 2 (HN)$^{\ddagger}$ & 0.95 & 0.73 & 0.9 & 0.87 & 0.98 \\
\bottomrule
\end{tabular}}
% 在表格下方添加注释，解释缩写含义
\begin{flushleft}
\footnotesize
\textit{Note:} AP = Ambiguous Pairs, HN = Hard Negatives.
\end{flushleft}

\caption{Comparison of ASR across different domains.}

\label{tab:stage_comparison_aligned}

\end{table}

\subsection{Analysis of Hard Negatives}
\label{sec:appendix_hard_negative}

We theoretically justify why constructing hard negative pairs is superior to random or easy negatives.

\begin{proposition}[Optimization Consistency \& Gradient Efficiency]
\label{prop:hard_negative}
Constructing preference pairs with a maximal detector score gap $\Delta D = D(y_l) - D(y_w) \approx 1$ is necessary to ensure consistency between empirical labels and the Bradley-Terry model, and to maximize the optimization gradient.
\end{proposition}

\begin{proof}
We define the evasion-oriented implicit reward as $r(y) = -C \cdot D(y)$ (where $C > 0$), meaning higher detection scores yield lower rewards. Substituting this into the Bradley-Terry model:
\begin{equation}
\begin{split}
    P_{\text{BT}}(y_w \succ y_l) &= \sigma(r(y_w) - r(y_l)) \\
    &= \sigma\big(-C D_w - (-C D_l)\big) \\
    &= \sigma\big(C \cdot \underbrace{(D_l - D_w)}_{\Delta D}\big),
\end{split}
\end{equation}
where $\sigma$ denotes the logistic sigmoid function. This derivation confirms that maximizing the probability requires maximizing the gap $\Delta D$. Table~\ref{tab:stage_comparison_aligned} confirms that employing hard negative samples yields the most significant improvement in ASR.

\begin{itemize}
    \item \textbf{Ambiguous Pairs ($\Delta D \approx 0$):} If $y_l$ evades detection ($D_l \approx 0$), then $\Delta D \approx 0$ and $P_{\text{BT}} \approx \sigma(0) = 0.5$. This contradicts the deterministic hard label ($y_w \succ y_l$), introducing label noise.
    \item \textbf{Hard Negatives ($\Delta D \approx 1$):} With $D_l \approx 1$ and $D_w \approx 0$, we have $P_{\text{BT}} \approx \sigma(C) \to 1$. This aligns with the empirical label.
\end{itemize}

The gradient of the DPO loss $\mathcal{L}_{\text{DPO}}$ in Eq.~\ref{eq:dpo_loss} with respect to model parameters $\theta$ is explicitly scaled by a weighting term $w(\theta)$:
\begin{equation}
    \nabla_\theta \mathcal{L}_{\text{DPO}} = -\mathbb{E} \left[ \underbrace{\sigma(\hat{r}_l - \hat{r}_w)}_{w(\theta)} \nabla_\theta \log \frac{\pi_\theta(y_w|x)}{\pi_\theta(y_l|x)} \right],
\end{equation}
where $\hat{r}_l$ and $\hat{r}_w$ denote the implicit rewards $\hat{r}_\theta(y_l)$ and $\hat{r}_\theta(y_w)$, respectively. For hard negatives, the model initially assigns high probability to the machine artifact $y_l$ (i.e., it confidently errs), leading to an inverted reward estimate where $\hat{r}_l > \hat{r}_w$. Consequently, the weighting term $w(\theta) > 0.5$, generates a strong error-correction signal. Conversely, easy negatives result in $\hat{r}_l \ll \hat{r}_w$, causing $w(\theta) \to 0$ and leading to vanishing gradients.

\end{proof}

\subsection{The Necessity of Style-SFT}
\label{sec:appendix_initialization}

Finally, we prove the necessity of Style-SFT (Stage 2) for effective DPO alignment.

\begin{theorem}[Support Constraint]
Let $\mathcal{Y}_{H} = \{ \mathbf{y} \mid D(\mathbf{y}) < \tau \}$ denote the target human-like region. If the reference policy $\pi_{\text{ref}}$ assigns negligible mass to $\mathcal{Y}_{H}$, the DPO-optimized policy $\pi^*$ cannot converge to $\mathcal{Y}_{H}$ regardless of the reward magnitude.
\end{theorem}

\begin{proof}
Recalling the closed-form optimal policy derived in Eq.~\ref{eq:dpo_optimal}. This multiplicative relationship implies that $\pi^*$ is strictly bounded by the support of $\pi_{\text{ref}}$. We analyze two initialization scenarios:
\begin{itemize}
    \item \textbf{Weak Initialization (Zero Support):} Without style injection, standard LLMs often exhibit a domain gap, such that $\forall \mathbf{y} \in \mathcal{Y}_{H}, \pi_{\text{ref}}(\mathbf{y}|\mathbf{x}) \approx 0$. Consequently, even if the implicit reward $r$ is maximized for human text, the posterior probability $\pi^*(\mathbf{y}|\mathbf{x}) \approx 0 \cdot \exp(\dots) = 0$. The model faces a ``cold start'' problem where high-reward regions are theoretically unreachable.
    \item \textbf{Style-SFT Initialization (Warm Start):} Stage 2 explicitly shifts probability mass towards the target style, ensuring $\pi_{\text{SFT}}(\mathbf{y}|\mathbf{x}) > \epsilon$, where $\epsilon$ is a minimal existence threshold. Crucially, a larger initial probability $\pi_{\text{SFT}}$ minimizes the distribution shift required to reach $\pi^*$, thereby reducing the optimization effort needed to satisfy the exponential scaling.
\end{itemize}

Thus, Stage 2 is a mathematical prerequisite to ensure the target region lies within the feasible search space of DPO.
\end{proof}

\section{Experimental Details}
\label{app:exp_details}

\noindent \textbf{Datasets.} 
We conduct experiments on two diverse benchmarks:
(1) MGTBench \citep{he2024mgtbench}, which comprises diverse writing styles: Essays (student argumentations), WP (creative fiction from Reddit WritingPrompts), and Reuters (financial and global news). Its AI samples are generated by seven LLMs (e.g., GPT-3.5, Llama, and Claude) using task-specific prompts to mimic human writing.
(2) MGT-Academic \citep{liu2025generalization}, covering formal scholarly discourse across: STEM (scientific and technical papers from ArXiv), Social Science (encyclopedic entries from Wikipedia), and Humanity (literary books from Project Gutenberg). Advanced LLMs (e.g., GPT-4, Gemini) generate these counterparts via academic paraphrasing or completion while preserving formal discourse.

To ensure rigorous evaluation, each subset is strictly partitioned into four splits—detector training, validation, shadow training, and testing—with a ratio of 3:1:1:1. The first two are used for fine-tuning the target RoBERTa detector, while the shadow set serves to train surrogate models for white-box baselines. Following \citet{fang-etal-2025-language}, we randomly sample 100 machine-generated instances per domain for the final testing.

% \noindent \textbf{Target Detectors.} We evaluate evasion performance against five black-box detectors: (1) Open-source: RoBERTa is fine-tuned on the training split of each subset, serving as a robust supervised baseline as identified by \citet{zheng2025th}. Binoculars \citep{hans2024spotting} is a zero-shot detector that requires no additional training or fine-tuning. SCRN \citep{huang-etal-2024-ai} is also utilized without further training, where we directly employ the official pre-trained weights released in the original implementation. 
% (2) Commercial: Following \citet{jianwang2024empirical}, we include Writer\footnote{\url{https://writer.com/ai-content-detector/}} and Scribbr\footnote{\url{https://www.scribbr.com/ai-detector/}}.
% These detectors represent real-world black-box scenarios, where access is strictly limited to an Oracle providing confidence scores or labels.

\noindent \textbf{Target Detectors.} 
We evaluate evasion performance against five black-box detectors. (1) Open-source: RoBERTa is fine-tuned on the training split of each subset, serving as a robust supervised baseline as identified by \citet{zheng2025th}. Binoculars \citep{hans2024spotting} is a zero-shot detector requiring no additional training. SCRN \citep{huang-etal-2024-ai} is used with its official pre-trained weights. (2) Commercial: Following \citet{jianwang2024empirical}, we include Writer\footnote{\url{https://writer.com/ai-content-detector/}} and Scribbr\footnote{\url{https://www.scribbr.com/ai-detector/}}. All detectors are evaluated under black-box settings, with access limited to confidence scores or binary labels.

\noindent \textbf{Baselines.}
% We benchmark MASH against 11 state-of-the-art evasion methods, categorized into perturbation-, prompt-, and paraphrase-based attacks. Perturbation-based methods include DeepWordBug \citep{gao2018black}, TextBugger \citep{li2018textbugger}, TextFooler \citep{jin2020bert}, Charmer \citep{Abad2024Charmer}, and HMGC \citep{Zhou2024_COLING}. To maintain a strict black-box setting, HMGC and GradEscape \citep{meng2025gradescape} are implemented by training auxiliary RoBERTa detectors on shadow training data. PromptAttack \citep{xu2024an} serves as the representative prompt-based attack, which guides models to generate evasion-prone text through optimized instructions.
We benchmark MASH against 11 state-of-the-art evasion methods, categorized into perturbation-, prompt-, and paraphrase-based attacks. Perturbation-based methods include DeepWordBug, TextBugger, TextFooler, Charmer, and HMGC. PromptAttack is the prompt-based baseline, which guides LLMs to generate evasion-prone text through optimized instructions. Paraphrase-based methods comprise DIPPER, Toblend, CoPA, GradEscape, and DPO-Evader. To maintain a strict black-box setting, HMGC and GradEscape are implemented by training auxiliary RoBERTa detectors on shadow training data.

% Paraphrase-based baselines comprise DIPPER \citep{krishna2023paraphrasing}, Toblend \citep{huang2024toblend}, CoPA \citep{fang-etal-2025-language}, GradEscape, and DPO-Evader \citep{nicks2023language}.
For DIPPER, we configure both lexical and order diversity to 40. In the Toblend evaluation, we employ GPT2-XL and OPT-2.7B as the source generators. Due to the absence of original generation prompts in the MGT-Academic benchmark, Toblend’s assessment is restricted to the MGTBench dataset. CoPA is implemented using Qwen-7B as the backbone paraphraser.
For DPO-Evader, we adopt Qwen-0.5B as the paraphraser to accommodate GPU memory constraints during training. Empirically, we observed that the 0.5B model achieves higher attack success rates than larger variants (e.g., 3B+), suggesting that smaller models may introduce more textual noise conducive to distribution alignment. All baseline implementations and hyperparameter configurations follow their respective original protocols unless otherwise specified.

\noindent \textbf{Metrics.}
Attack Success Rate (ASR) is the proportion of polished texts $\{x_{adv}^{(i)}\}_{i=1}^N$ classified as ``human-written'' by detector $D$:
\begin{equation}
\text{ASR} = \frac{1}{N} \sum_{i=1}^{N} \mathbbm{1}(D(x^{(i)}_{adv}) < \tau),
\end{equation}
where $\mathbbm{1}(\cdot)$ is the indicator function and $\tau$ is the decision threshold.

Following \citet{Zhou2024_COLING}, we assess fluency via Perplexity (PPL), calculated using Pythia-2.8b. A lower PPL indicates superior text fluency. For a sequence $x$ of length $T$, PPL is computed as the exponential of the average negative log-likelihood:
\begin{equation}
\text{PPL}(x) = \exp \left( -\frac{1}{T} \sum_{t=1}^{T} \log P(x_t \mid x_{<t}) \right).
\end{equation}

To evaluate semantic consistency between $x_{ai}$ and $x_{adv}$, we employ BERTScore \citep{zhang2019bertscore}. BERTScore computes token-level similarity using contextual embeddings and aggregates them into precision, recall, and F1 metrics. Following standard practice, we report the F1-score to provide a robust assessment of semantic integrity.

Distinct from PPL, we employ GRUEN \citep{zhu2020gruen} for a holistic, reference-free linguistic quality assessment. Higher GRUEN values indicate better overall linguistic quality.

\noindent \textbf{Implementation and Training Details.}
All experiments were conducted on a single NVIDIA RTX 3090 GPU, with a complete training cycle per domain lasting approximately 7.5 hours.
% During the Style-SFT stage, we jointly optimize the reconstruction loss ($\mathcal{L}_{recon}$) and the style transfer loss ($\mathcal{L}_{trans}$). Empirically, the balancing coefficient $\lambda$ is set to $1.0$.

During the Style-SFT stage, We jointly optimize $\mathcal{L}_{recon}$ and $\mathcal{L}_{trans}$ with $\lambda = 0.5$.
We employ the AdamW optimizer with a learning rate of $2 \times 10^{-5}$ and a batch size of 8. The model is trained for 50 epochs with early stopping.

During the DPO phase, the learning rate is adjusted to $5 \times 10^{-6}$ for 5 epochs. To accommodate GPU memory constraints, we use a per-device batch size of 2 combined with 8 gradient accumulation steps.
% The maximum sequence length is set to 512 tokens for SFT and 1024 tokens for DPO.
For inference, we utilize beam search with a beam size of 4.

For the Reuters domain in MGTBench, we utilize the \texttt{vblagoje/cc\_news} dataset, while for all other domains, we employ the \texttt{dmitva/human\_ai\_generated\_text} dataset.
% To synthesize semantically equivalent AI-generated counterparts, we leverage an ensemble of large language models, including Qwen2.5-3B-Instruct, Meta-Llama-3-8B-Instruct, and GPT-2, to paraphrase the original human sources.
We paraphrase human sources using Qwen2.5-3B-Instruct, Llama-3-8B-Instruct, and GPT-2.
The vLLM framework is adopted to accelerate the inference process. This pipeline yielded approximately 6,000 validated pairs, achieving a label-flip success rate (transitioning from human-written to AI-generated) of approximately 95\%.

\begin{figure}[h]
\centering
\includegraphics[width=\columnwidth]{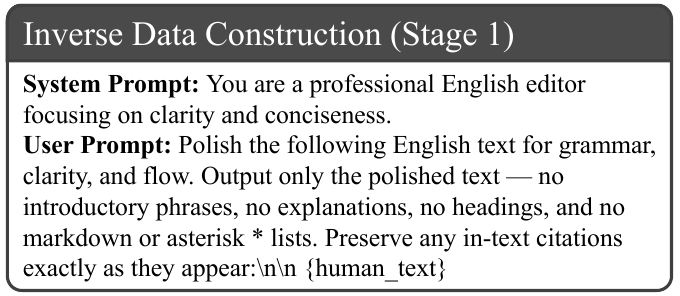}
\end{figure}

In Stage 4, sentence-level polish candidates are generated by querying the GPT-5 API. This refinement module ensures grammatical fluency while strictly preserving the evasion status of the text.

\begin{figure}[h]
\centering
\includegraphics[width=\columnwidth]{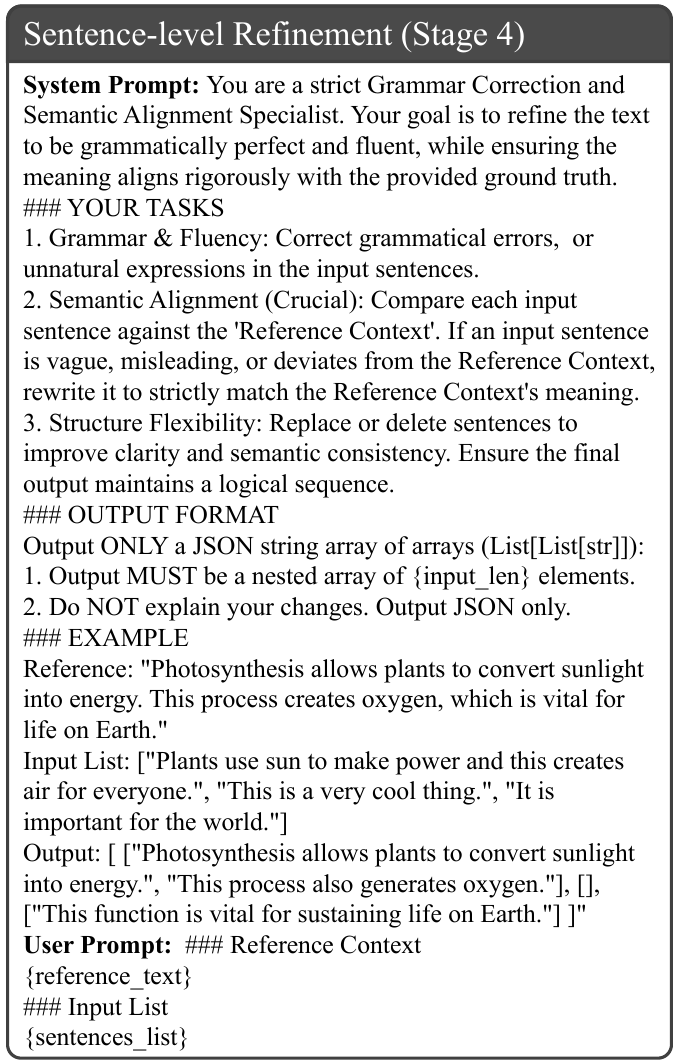}
\end{figure}

\section{Additional Experimental Results}
\label{sec:appendix_additional_results}

\begin{figure}[t]
\centering
\includegraphics[width=\columnwidth]{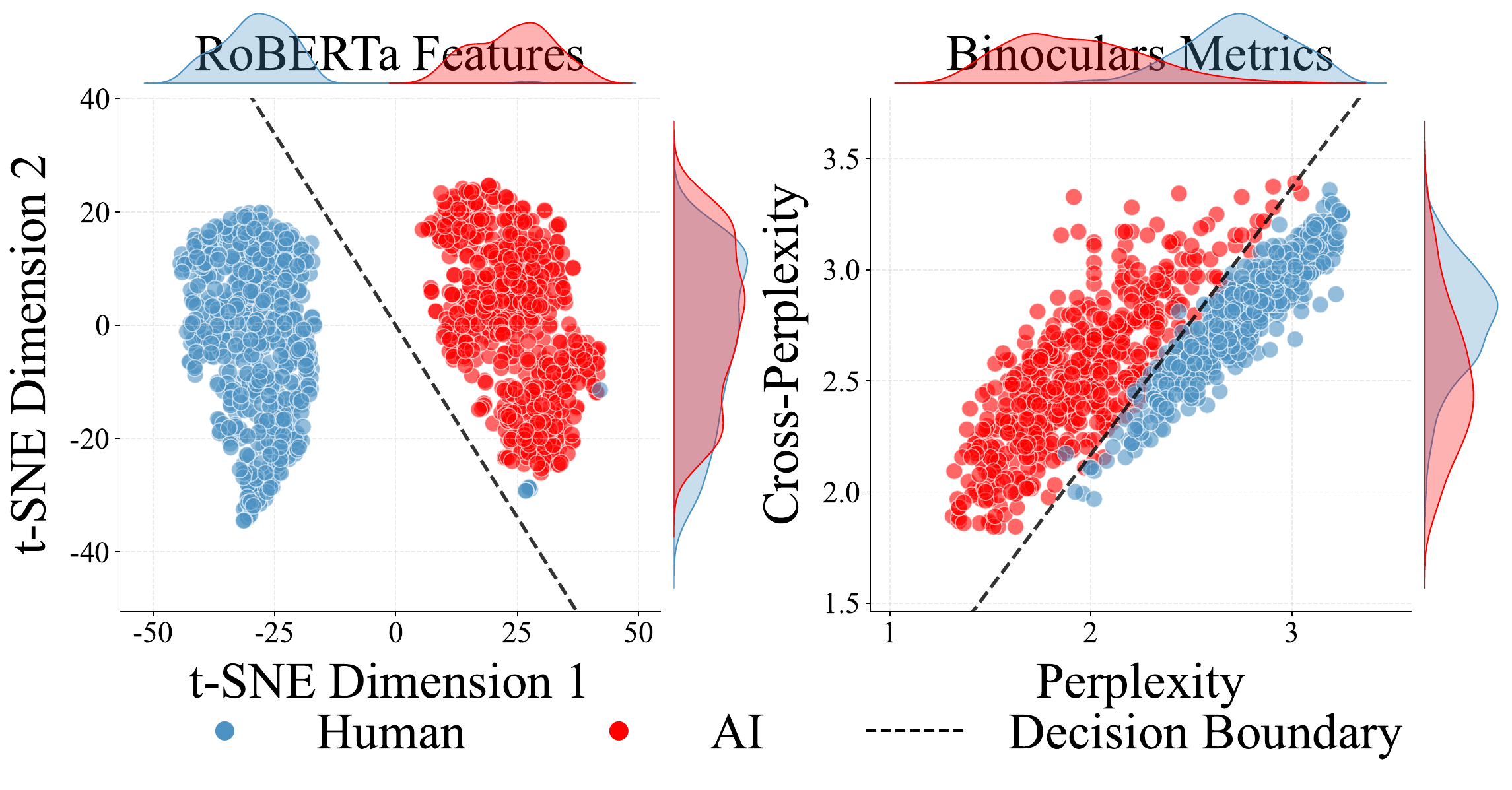}
\caption{Visualization of the feature distribution gap between AI-generated and human-written text. The left plot shows the t-SNE projection of semantic features from the RoBERTa detector. The right plot illustrates the statistical features from the Binoculars detector. Marginal KDEs highlight dense regions.}
\label{fig:motivation}
\end{figure}

\noindent \textbf{Visualization of Decision Boundaries.} The superior performance of current AI-generated text detectors is fundamentally rooted in the distinct distribution discrepancy between machine-generated and human-written text. 
As shown in Figure~\ref{fig:motivation}, based on 1,200 samples from the Essay dataset, distinct separation is observable in both the deep semantic feature space extracted by supervised detectors and the statistical feature space utilized by zero-shot detectors. 
Consequently, we posit that effective evasion should not be merely adversarial noise injection against a detector, but rather a style transfer task.

\begin{figure}[t]
\includegraphics[width=\columnwidth]{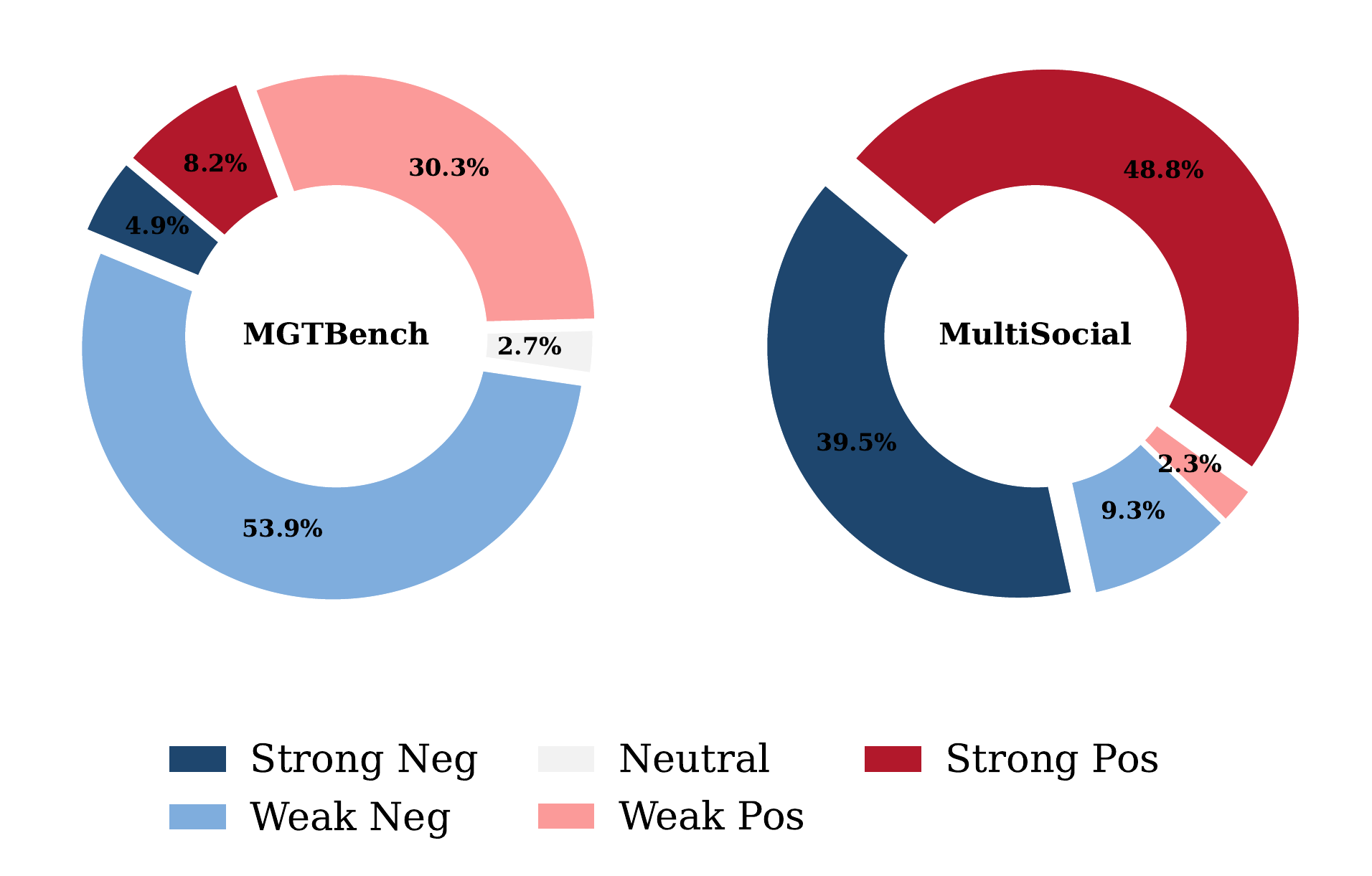}
\caption{Attribution distribution of token contributions on MultiSocial and MGTBench, calculated via gradient-weighted embedding attribution. The left shows sparsely distributed decision influence; the right shows concentrated high-contribution tokens.}
\label{fig:short_and_long}
\end{figure}

\noindent \textbf{Token-Level Detection Contribution.} We hypothesize that the effectiveness of perturbation-based attacks correlates with token contribution density. For validation, we analyzed 100 randomly sampled instances from MultiSocial \citep{macko2025multisocial} and MGTBench. As illustrated in Figure~\ref{fig:short_and_long}, when high-contribution tokens are concentrated, perturbation-based attacks achieve a high hit rate, easily flipping classifications through local edits. However, when the decision influence is sparsely distributed across a vast ``neutral zone,'' these attacks predominantly target non-critical regions, failing to induce the substantive shift in global feature distribution necessary to evade robust detectors.

\begin{table}[t]
\centering
\small

\resizebox{\columnwidth}{!}{
\begin{tabular}{lcccc}
\toprule
\textbf{Dataset} & \textbf{Precision} & \textbf{Recall} & \textbf{F1 Score} & \textbf{Accuracy} \\
\midrule
\multicolumn{5}{c}{\textit{Target Detector: Writer}} \\
\midrule
Essay & 1.0000 & 0.9800 & 0.9899 & 0.9900 \\
Reuters & 1.0000 & 0.8571 & 0.9231 & 0.9293 \\
WP & 0.9800 & 0.9800 & 0.9800 & 0.9800 \\
Humanities & 0.9091 & 0.6122 & 0.7317 & 0.7778 \\
Social & 0.9697 & 0.6400 & 0.7711 & 0.8100 \\
STEM & 1.0000 & 0.5800 & 0.7342 & 0.7900 \\
\midrule
\multicolumn{5}{c}{\textit{Target Detector: Scribbr}} \\
\midrule
Essay & 1.0000 & 0.5400 & 0.7013 & 0.7700 \\
Reuters & 1.0000 & 0.3600 & 0.5294 & 0.6800 \\
WP & 1.0000 & 0.5000 & 0.6667 & 0.7500 \\
Humanities & 1.0000 & 0.3400 & 0.5075 & 0.6700 \\
Social & 1.0000 & 0.4200 & 0.5915 & 0.7100 \\
STEM & 1.0000 & 0.4800 & 0.6486 & 0.7400 \\
\bottomrule
\end{tabular}
}

\caption{Performance metrics of commercial detectors.}
\label{tab:commercial_performance}
\end{table}

% \noindent \textbf{Commercial Detector Baselines.} To establish a rigorous benchmark, we first assessed the detection capabilities of Writer and Scribbr across all six domains. As shown in Table~\ref{tab:commercial_performance}, both systems exhibit high precision, with Writer demonstrating superior recall and peaking at 0.99 accuracy on the Essay dataset, confirming their validity as robust baselines. However, given the prohibitive query costs of commercial APIs, performing a full-scale evasion attack is impractical. Therefore, we adopted a strategic evaluation protocol \citep{meng2025gradescape}: we targeted the two most challenging domains—Essay and WP—where these detectors exhibited their strongest performance. This approach allows us to rigorously stress-test MASH against the most robust defenses while maintaining a cost-effective experimental budget (under \$20).

\noindent \textbf{Commercial Detector Baselines.} To establish a rigorous benchmark, we assessed the detection capabilities of Writer and Scribbr across all six domains. As shown in Table~\ref{tab:commercial_performance}, we report four standard classification metrics: Precision (the proportion of detected AI texts that are truly AI-generated), Recall (the proportion of AI-generated texts correctly identified), F1 Score (the harmonic mean of Precision and Recall), and Accuracy (the proportion of correctly classified samples). Both systems exhibit high precision, with Writer demonstrating superior recall and achieving 0.99 accuracy on the Essay dataset. These results confirm their validity as robust baselines for evaluating evasion attacks. The total query cost for all commercial detector experiments was under \$20.

\begin{figure}[h]
\centering
\includegraphics[width=\columnwidth]{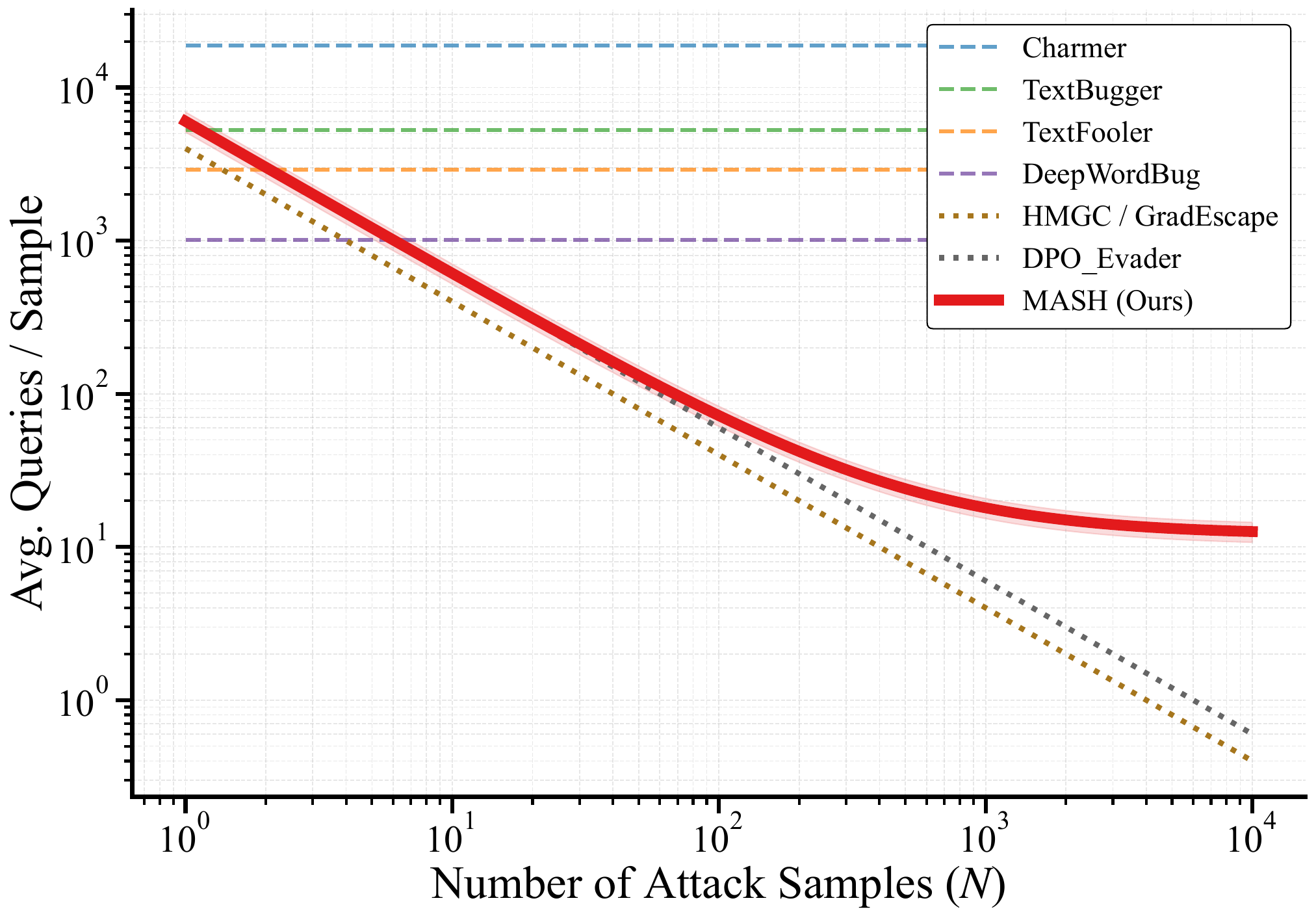}
\caption{Amortized query cost per sample. MASH's initial training overhead is rapidly diluted as attack volume grows, achieving significant efficiency gains over perturbation-based methods at scale.}
\label{fig:query_efficiency}
\end{figure}

\noindent \textbf{Amortized Query Cost Analysis.} 
Figure~\ref{fig:query_efficiency} shows the average query cost per sample as a function of the total attack volume ($N$). Perturbation-based methods (e.g., Charmer, TextFooler) require iterative optimization for each input, resulting in a constant and high per-sample query cost regardless of scale. In contrast, MASH incurs a fixed upfront cost for training the paraphraser (Stages 2 and 3), but requires no additional queries during inference. As $N$ increases, this training cost is distributed across more samples, leading to a sharp decrease in per-sample overhead. For large-scale attacks ($N > 10^4$), the average cost of MASH becomes negligible, achieving orders-of-magnitude better efficiency than query-based methods. This confirms that MASH is highly cost-effective for high-volume evasion scenarios.

\section{Defensive Analysis}
\label{sec:appendix_defense}

\begin{table}[t]
\centering
\small

% 调整表格宽度以适应栏宽
\resizebox{\columnwidth}{!}{
\begin{tabular}{lcccc}
\toprule
\multirow{2}{*}{\textbf{Dataset}} & \multicolumn{2}{c}{\textbf{w/o Training}} & \multicolumn{2}{c}{\textbf{w Training}} \\
\cmidrule(lr){2-3} \cmidrule(lr){4-5}
 & \textbf{Accuracy} & \textbf{ASR} & \textbf{Accuracy} & \textbf{ASR} \\
\midrule
Essay & 0.9974 & 0.95 & 0.3059 & 0.01 \\
Reuters & 0.9988 & 0.73 & 0.9988 & 0 \\
WP & 0.9912 & 0.90 & 0.8731 & 0.01 \\
Humanity & 0.9943 & 0.87 & 0.9543 & 0.12 \\
Social Science & 0.9898 & 0.98 & 0.838 & 0 \\
STEM & 0.9963 & 1 & 0.878 & 0 \\
\bottomrule
\end{tabular}
}
\caption{Performance comparison on different datasets with and without adversarial training.}
\label{tab:adversarial_training_results}
\end{table}

While MASH effectively evades existing detectors, we further investigate potential mitigation strategies against such style-based humanization attacks. 
We propose an adversarial training approach in which the target detector is fine-tuned on a mixture of original human texts and MASH-polished adversarial examples.
We use RoBERTa as the base detector and construct a balanced training set of 1,000 human-written texts and 1,000 MASH-rewritten machine-generated texts.

The results in Table~\ref{tab:adversarial_training_results} reveal a clear robustness-accuracy trade-off.
Adversarial training substantially improves the detector's robustness against MASH: the ASR drops to nearly 0.00 across most domains (e.g., STEM, Reuters), demonstrating that the detector can learn to recognize humanized styles.
However, this defense improves detection of MASH-generated texts at the cost of reduced accuracy on unattacked samples.

\section{Visualization of Rewritten Texts}
\label{sec:vis_rewritten}

Figure~\ref{fig:tsne} shows the t-SNE projection of the feature 
representations from the RoBERTa detector. The plots reveal a 
clear distribution shift: original AI texts (red) are well 
separated from human texts (blue), whereas MASH-polished texts 
(green) successfully migrate into the human cluster, validating 
the effectiveness of our style transfer.

\begin{figure*}[t]
  \centering  % 图片居中
  \includegraphics[width=1\linewidth]{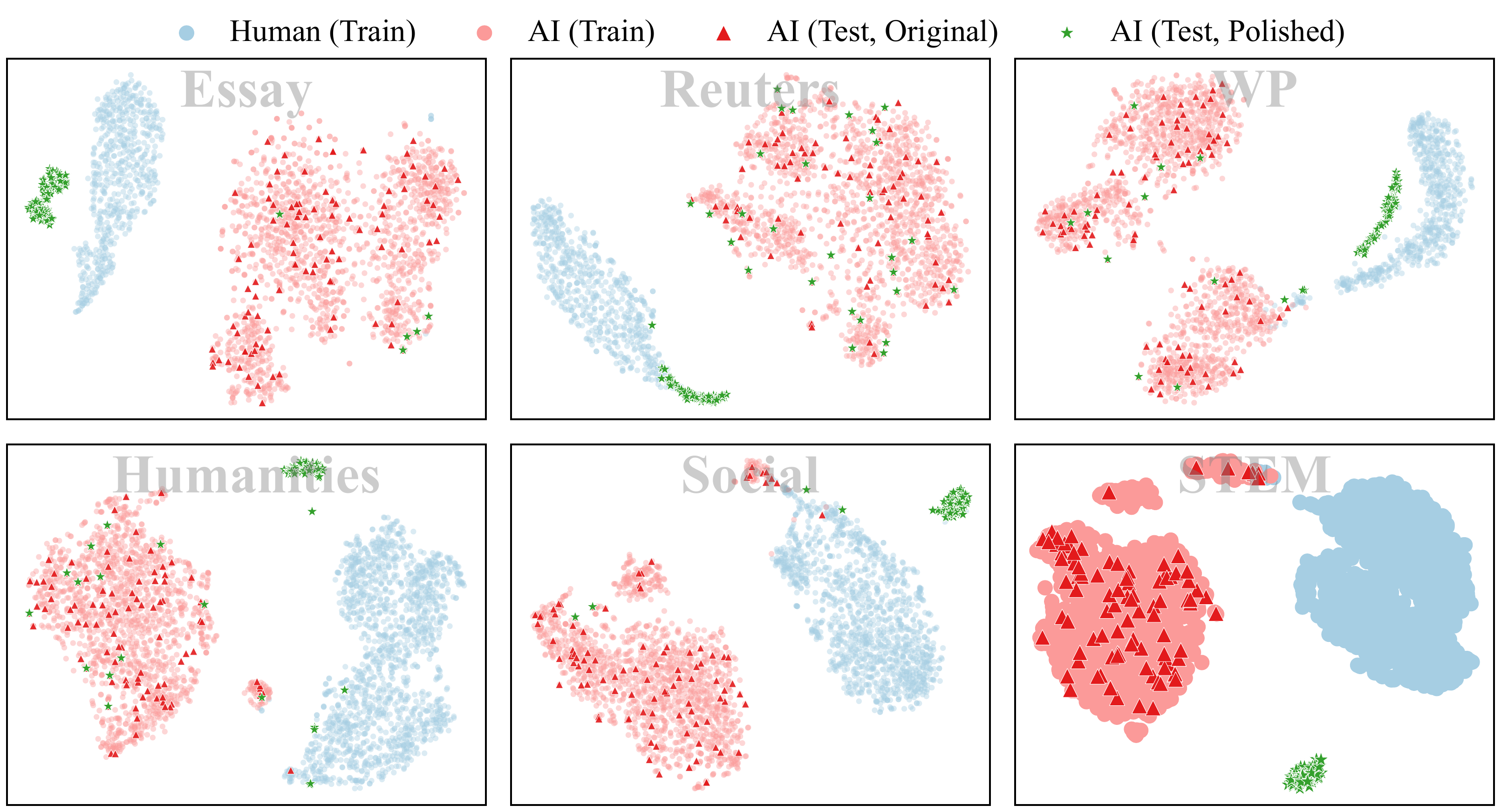} 
  \caption{t-SNE visualization of feature representations extracted from the RoBERTa detector across six domains.}
  \label{fig:tsne} %以此标签引用图片
\end{figure*}

\section{Additional Ablation Studies}
\label{sec:add_abla}

\begin{table}[t]
  \centering

  \small
  \begin{tabular}{llccc}
    \toprule
    \textbf{Dataset} & \textbf{Model} & \textbf{PPL}$\downarrow$ & \textbf{GRUEN}$\uparrow$ & \textbf{BS}$\uparrow$ \\
    \midrule
    \multirow{4}{*}{Essay} 
      & Original   & 25.10 & 0.371 & 0.830 \\
      & Llama-3-8B & 21.01 & 0.362 & 0.856 \\
      & Qwen2.5-7B & 20.70 & 0.587 & 0.874 \\
      & \cellcolor{gray!20}\textbf{GPT-5 (Ours)} & \cellcolor{gray!20}\textbf{18.90} & \cellcolor{gray!20}\textbf{0.661} & \cellcolor{gray!20}\textbf{0.900} \\
    \midrule
    \multirow{4}{*}{Reuters} 
      & Original   & \textbf{8.42} & 0.678 & \textbf{0.904} \\
      & Llama-3-8B & 9.09 & 0.679 & 0.902 \\
      & Qwen2.5-7B & 9.09 & 0.679 & 0.902 \\
      & \cellcolor{gray!20}\textbf{GPT-5 (Ours)} & \cellcolor{gray!20}9.09 & \cellcolor{gray!20}\textbf{0.679} & \cellcolor{gray!20}0.902 \\
    \midrule
    \multirow{4}{*}{WP} 
      & Original   & 24.68 & 0.475 & 0.822 \\
      & Llama-3-8B & \textbf{19.67} & 0.493 & 0.831 \\
      & Qwen2.5-7B & 19.74 & 0.585 & 0.838 \\
      & \cellcolor{gray!20}\textbf{GPT-5 (Ours)} & \cellcolor{gray!20}20.80 & \cellcolor{gray!20}\textbf{0.647} & \cellcolor{gray!20}\textbf{0.897} \\
    \midrule
    \multirow{4}{*}{Humanity} 
      & Original   & 30.81 & 0.365 & 0.819 \\
      & Llama-3-8B & \textbf{19.81} & 0.564 & 0.851 \\
      & Qwen2.5-7B & 21.46 & 0.631 & 0.862 \\
      & \cellcolor{gray!20}\textbf{GPT-5 (Ours)} & \cellcolor{gray!20}20.67 & \cellcolor{gray!20}\textbf{0.640} & \cellcolor{gray!20}\textbf{0.905} \\
    \midrule
    \multirow{4}{*}{Social} 
      & Original   & 22.55 & 0.457 & 0.800 \\
      & Llama-3-8B & \textbf{15.80} & 0.575 & 0.843 \\
      & Qwen2.5-7B & 18.63 & 0.699 & 0.854 \\
      & \cellcolor{gray!20}\textbf{GPT-5 (Ours)} & \cellcolor{gray!20}17.54 & \cellcolor{gray!20}\textbf{0.729} & \cellcolor{gray!20}\textbf{0.899} \\
    \midrule
    \multirow{4}{*}{STEM} 
      & Original   & 22.89 & 0.463 & 0.783 \\
      & Llama-3-8B & \textbf{18.67} & 0.500 & 0.815 \\
      & Qwen2.5-7B & 21.82 & 0.606 & 0.810 \\
      & \cellcolor{gray!20}\textbf{GPT-5 (Ours)} & \cellcolor{gray!20}20.08 & \cellcolor{gray!20}\textbf{0.643} & \cellcolor{gray!20}\textbf{0.819} \\
    \bottomrule
  \end{tabular}

  \caption{Ablation study on Stage 4 refinement models.}

    \label{tab:roberta_quality_comparison}
\end{table}

\noindent \textbf{Alternative Refinement Models.}
Table~\ref{tab:roberta_quality_comparison} compares text quality across 
different refinement models to evaluate open-source alternatives for 
Stage 4. We assess Llama-3-8B-Instruct and Qwen2.5-7B-Instruct against 
our default GPT-5 refinement. Results show that both open-source models 
improve over original outputs, with Qwen2.5-7B achieving competitive 
GRUEN scores. However, GPT-5 consistently yields superior BERTScore, 
indicating better semantic preservation. These findings suggest that 
open-source refiners offer viable alternatives under cost constraints, 
though with a modest trade-off in semantic fidelity.

\begin{table*}[t]
  \centering

  \resizebox{\textwidth}{!}{%
  \begin{tabular}{l|l|cccc|cccc|cccc}
    \toprule
    \multirow{2}{*}{\textbf{Detector}} & \multirow{2}{*}{\textbf{Method Variant}} & \multicolumn{4}{c|}{\textbf{MGTBench-Essay}} & \multicolumn{4}{c|}{\textbf{MGTBench-Reuters}} & \multicolumn{4}{c}{\textbf{MGTBench-WP}} \\
    \cmidrule(lr){3-6} \cmidrule(lr){7-10} \cmidrule(lr){11-14}
     & & \textbf{ASR} & \textbf{PPL} & \textbf{GRUEN} & \textbf{BS} & \textbf{ASR} & \textbf{PPL} & \textbf{GRUEN} & \textbf{BS} & \textbf{ASR} & \textbf{PPL} & \textbf{GRUEN} & \textbf{BS} \\
    \midrule
    \multirow{3}{*}{RoBERTa} 
     & Style-SFT & 0.83 & 25.10 & 0.5033 & 0.8823 & 0.09 & 8.58 & \textbf{0.6968} & \textbf{0.9930} & 0.68 & 19.79 & 0.5307 & 0.8781 \\
     & + DPO & \textbf{0.95} & 25.10 & 0.3711 & 0.8301 & \textbf{0.73} & 8.42 & 0.6779 & 0.9044 & \textbf{0.90} & 24.68 & 0.4748 & 0.8222 \\
     \rowcolor{gray!20} & + Refine (Ours) & \textbf{0.95} & \textbf{18.90} & \textbf{0.6614} & 0.9004 & \textbf{0.73} & \textbf{9.09} & 0.6790 & 0.9015 & \textbf{0.90} & \textbf{20.80} & \textbf{0.6471} & \textbf{0.8974} \\
    \midrule
    \multirow{3}{*}{SCRN} 
     & Style-SFT & 0.33 & 22.40 & 0.4615 & 0.8881 & 0.62 & 26.12 & 0.4912 & 0.8682 & 0.39 & 19.80 & 0.5727 & 0.9384 \\
     & + DPO & 0.56 & 23.11 & 0.4611 & 0.8679 & \textbf{0.83} & 26.26 & 0.4590 & 0.8496 & \textbf{0.81} & 20.34 & 0.5730 & 0.9258 \\
     \rowcolor{gray!20} & + Refine (Ours) & \textbf{0.56} & \textbf{13.30} & \textbf{0.6984} & \textbf{0.8951} & \textbf{0.83} & \textbf{14.01} & \textbf{0.6951} & \textbf{0.8964} & \textbf{0.81} & \textbf{12.70} & \textbf{0.6811} & \textbf{0.9433} \\
    \midrule
    \multirow{3}{*}{Binoculars} 
     & Style-SFT & 0.88 & 13.73 & 0.5143 & 0.8682 & 0.92 & 13.90 & 0.5293 & 0.8472 & 0.82 & 13.69 & 0.5889 & 0.8982 \\
     & + DPO & \textbf{0.94} & 13.24 & 0.4914 & 0.8540 & \textbf{0.95} & 13.32 & 0.4432 & 0.8335 & \textbf{0.85} & 13.58 & 0.5477 & 0.8809 \\
     \rowcolor{gray!20} & + Refine (Ours) & \textbf{0.94} & \textbf{11.39} & \textbf{0.6755} & \textbf{0.9077} & \textbf{0.95} & \textbf{12.91} & \textbf{0.7122} & \textbf{0.9037} & \textbf{0.85} & \textbf{12.05} & \textbf{0.6883} & \textbf{0.9286} \\
    \midrule
    \midrule
    \multirow{2}{*}{\textbf{Detector}} & \multirow{2}{*}{\textbf{Method Variant}} & \multicolumn{4}{c|}{\textbf{MGT-Academic-Humanity}} & \multicolumn{4}{c|}{\textbf{MGT-Academic-Social Science}} & \multicolumn{4}{c}{\textbf{MGT-Academic-STEM}} \\
    \cmidrule(lr){3-6} \cmidrule(lr){7-10} \cmidrule(lr){11-14}
     & & \textbf{ASR} & \textbf{PPL} & \textbf{GRUEN} & \textbf{BS} & \textbf{ASR} & \textbf{PPL} & \textbf{GRUEN} & \textbf{BS} & \textbf{ASR} & \textbf{PPL} & \textbf{GRUEN} & \textbf{BS} \\
    \midrule
    \multirow{3}{*}{RoBERTa} 
     & Style-SFT & 0.68 & 26.61 & 0.4366 & 0.8685 & 0.92 & 23.03 & 0.5738 & 0.8204 & \textbf{1.00} & 23.38 & 0.5221 & 0.7868 \\
     & + DPO & \textbf{0.87} & 30.81 & 0.3651 & 0.8193 & \textbf{0.98} & 22.55 & 0.4568 & 0.8004 & \textbf{1.00} & 22.89 & 0.4626 & 0.7833 \\
     \rowcolor{gray!20} & + Refine (Ours) & \textbf{0.87} & \textbf{20.67} & \textbf{0.6396} & \textbf{0.9048} & \textbf{0.98} & \textbf{17.54} & \textbf{0.7294} & \textbf{0.8993} & \textbf{1.00} & \textbf{20.08} & \textbf{0.6431} & \textbf{0.8194} \\
    \midrule
    \multirow{3}{*}{SCRN} 
     & Style-SFT & 0.70 & 31.77 & 0.5092 & 0.8955 & 0.65 & 27.24 & 0.4768 & 0.8816 & 0.75 & 38.57 & 0.4716 & 0.8721 \\
     & + DPO & \textbf{0.77} & 32.18 & 0.5072 & 0.8794 & \textbf{0.78} & 27.62 & 0.4742 & 0.8644 & \textbf{0.89} & 39.86 & 0.4374 & 0.8506 \\
     \rowcolor{gray!20} & + Refine (Ours) & \textbf{0.77} & \textbf{14.85} & \textbf{0.6762} & \textbf{0.9299} & \textbf{0.78} & \textbf{13.35} & \textbf{0.6626} & \textbf{0.9100} & \textbf{0.89} & \textbf{14.71} & \textbf{0.6488} & \textbf{0.9091} \\
    \midrule
    \multirow{3}{*}{Binoculars} 
     & Style-SFT & 0.92 & 15.16 & 0.5248 & 0.8454 & 0.95 & 16.02 & 0.5036 & 0.8480 & 0.94 & 17.60 & 0.4400 & 0.8280 \\
     & + DPO & \textbf{0.95} & 14.01 & 0.4184 & 0.8313 & \textbf{0.95} & 14.63 & 0.4445 & 0.8339 & \textbf{0.99} & 16.63 & 0.3508 & 0.8157 \\
     \rowcolor{gray!20} & + Refine (Ours) & \textbf{0.95} & \textbf{13.47} & \textbf{0.7025} & \textbf{0.9074} & \textbf{0.95} & \textbf{13.31} & \textbf{0.7014} & \textbf{0.9073} & \textbf{0.99} & \textbf{14.02} & \textbf{0.6434} & \textbf{0.8864} \\
    \bottomrule
  \end{tabular}%
  }

    \caption{Ablation study of MASH components across three detectors.}
  \label{tab:ablation_study}
  
\end{table*}

\noindent \textbf{Impact of MASH Components.}
Table~\ref{tab:ablation_study} evaluates the contributions of Style-SFT, 
DPO alignment, and inference-time refinement across three detectors and 
six datasets. The results demonstrate that the DPO module is critical 
for crossing detector decision boundaries, yielding a significant boost 
in ASR; however, this improvement often degrades textual fluency. The 
inference-time refinement module effectively mitigates this trade-off.
% By significantly reducing Perplexity (PPL) and enhancing GRUEN scores while preserving high ASR, this stage proves instrumental in restoring linguistic quality and ensuring semantic consistency.

\begin{table*}[t]
  \centering

  \resizebox{\textwidth}{!}{%
  \begin{tabular}{l|ccc|ccc|ccc}
    \toprule
    \multirow{2}{*}{\textbf{Method}} & \multicolumn{3}{c|}{\textbf{MGTBench-Essay}} & \multicolumn{3}{c|}{\textbf{MGTBench-Reuters}} & \multicolumn{3}{c}{\textbf{MGTBench-WP}} \\
    \cmidrule(lr){2-4} \cmidrule(lr){5-7} \cmidrule(lr){8-10}
     & \textbf{PPL} $\downarrow$ & \textbf{GRUEN} $\uparrow$ & \textbf{BS} $\uparrow$ & \textbf{PPL} $\downarrow$ & \textbf{GRUEN} $\uparrow$ & \textbf{BS} $\uparrow$ & \textbf{PPL} $\downarrow$ & \textbf{GRUEN} $\uparrow$ & \textbf{BS} $\uparrow$ \\
    \midrule
    DeepWordBug & 63.02 $\to$ 62.57 & 0.470 $\to$ 0.485 & 0.934 $\to$ 0.933 & 88.83 $\to$ 91.32 & 0.524 $\to$ 0.524 & 0.933 $\to$ 0.929 & 32.85 $\to$ 30.41 & 0.524 $\to$ 0.553 & 0.950 $\to$ 0.951 \\
    TextBugger & 40.67 $\to$ 33.32 & 0.457 $\to$ 0.516 & 0.919 $\to$ 0.932 & 84.02 $\to$ 88.84 & 0.404 $\to$ 0.404 & 0.893 $\to$ 0.890 & 34.24 $\to$ 24.15 & 0.494 $\to$ 0.532 & 0.932 $\to$ 0.945 \\
    TextFooler & 77.18 $\to$ 72.24 & 0.447 $\to$ 0.480 & 0.940 $\to$ 0.943 & 181.46 $\to$ 197.37 & 0.302 $\to$ 0.303 & 0.870 $\to$ 0.868 & 50.20 $\to$ 37.01 & 0.425 $\to$ 0.490 & 0.938 $\to$ 0.946 \\
    DIPPER & 44.78 $\to$ 44.76 & 0.648 $\to$ 0.651 & 0.903 $\to$ 0.906 & 11.07 $\to$ 11.07 & 0.718 $\to$ 0.718 & 0.910 $\to$ 0.910 & 14.11 $\to$ 14.09 & 0.577 $\to$ 0.583 & 0.901 $\to$ 0.902 \\
    Toblend & 7.03 $\to$ 7.07 & 0.378 $\to$ 0.378 & 0.753 $\to$ 0.752 & 6.88 $\to$ 6.88 & 0.558 $\to$ 0.558 & 0.724 $\to$ 0.724 & 6.59 $\to$ 6.65 & 0.481 $\to$ 0.487 & 0.731 $\to$ 0.737 \\
    PromptAttack & 11.21 $\to$ 11.99 & 0.778 $\to$ 0.772 & 0.916 $\to$ 0.915 & 10.55 $\to$ 11.41 & 0.791 $\to$ 0.791 & 0.927 $\to$ 0.924 & 10.55 $\to$ 11.27 & 0.726 $\to$ 0.726 & 0.925 $\to$ 0.921 \\
    Charmer & 38.00 $\to$ 37.40 & 0.553 $\to$ 0.555 & 0.966 $\to$ 0.963 & 114.57 $\to$ 115.62 & 0.583 $\to$ 0.584 & 0.962 $\to$ 0.960 & 17.26 $\to$ 17.02 & 0.579 $\to$ 0.566 & 0.976 $\to$ 0.971 \\
    DPO-Evader & 5.14 $\to$ 5.27 & 0.154 $\to$ 0.154 & 0.934 $\to$ 0.930 & 4.53 $\to$ 4.65 & 0.149 $\to$ 0.149 & 0.926 $\to$ 0.923 & 5.41 $\to$ 5.52 & 0.204 $\to$ 0.202 & 0.923 $\to$ 0.918 \\
    HMGC & 57.80 $\to$ 49.17 & 0.588 $\to$ 0.604 & 0.923 $\to$ 0.930 & 82.60 $\to$ 88.63 & 0.547 $\to$ 0.547 & 0.901 $\to$ 0.898 & 41.82 $\to$ 37.87 & 0.576 $\to$ 0.573 & 0.932 $\to$ 0.931 \\
    GradEscape & 17.72 $\to$ 16.76 & 0.135 $\to$ 0.191 & 0.904 $\to$ 0.911 & 74.96 $\to$ 84.20 & 0.481 $\to$ 0.482 & 0.814 $\to$ 0.815 & 13.49 $\to$ 13.59 & 0.652 $\to$ 0.648 & 0.987 $\to$ 0.981 \\
    CoPA & 29.42 $\to$ 30.33 & 0.689 $\to$ 0.691 & 0.875 $\to$ 0.876 & 41.64 $\to$ 41.97 & 0.694 $\to$ 0.694 & 0.877 $\to$ 0.876 & 25.47 $\to$ 23.23 & 0.633 $\to$ 0.649 & 0.881 $\to$ 0.891 \\
    \midrule
    \multirow{2}{*}{\textbf{Method}} & \multicolumn{3}{c|}{\textbf{MGT-Academic-Humanity}} & \multicolumn{3}{c|}{\textbf{MGT-Academic-Social Science}} & \multicolumn{3}{c}{\textbf{MGT-Academic-STEM}} \\
    \cmidrule(lr){2-4} \cmidrule(lr){5-7} \cmidrule(lr){8-10}
     & \textbf{PPL} $\downarrow$ & \textbf{GRUEN} $\uparrow$ & \textbf{BS} $\uparrow$ & \textbf{PPL} $\downarrow$ & \textbf{GRUEN} $\uparrow$ & \textbf{BS} $\uparrow$ & \textbf{PPL} $\downarrow$ & \textbf{GRUEN} $\uparrow$ & \textbf{BS} $\uparrow$ \\
    \midrule
    DeepWordBug & 30.09 $\to$ 24.03 & 0.487 $\to$ 0.517 & 0.952 $\to$ 0.959 & 32.59 $\to$ 22.18 & 0.541 $\to$ 0.600 & 0.951 $\to$ 0.963 & 22.64 $\to$ 19.88 & 0.466 $\to$ 0.489 & 0.954 $\to$ 0.958 \\
    TextBugger & 40.51 $\to$ 22.91 & 0.403 $\to$ 0.507 & 0.932 $\to$ 0.958 & 56.87 $\to$ 26.43 & 0.443 $\to$ 0.570 & 0.936 $\to$ 0.959 & 39.93 $\to$ 27.79 & 0.404 $\to$ 0.479 & 0.933 $\to$ 0.946 \\
    TextFooler & 72.33 $\to$ 31.04 & 0.350 $\to$ 0.494 & 0.929 $\to$ 0.955 & 136.89 $\to$ 47.22 & 0.332 $\to$ 0.525 & 0.912 $\to$ 0.952 & 81.24 $\to$ 45.71 & 0.319 $\to$ 0.448 & 0.926 $\to$ 0.946 \\
    DIPPER & 14.67 $\to$ 14.38 & 0.578 $\to$ 0.595 & 0.903 $\to$ 0.925 & 13.88 $\to$ 13.31 & 0.683 $\to$ 0.673 & 0.906 $\to$ 0.926 & 12.47 $\to$ 11.66 & 0.604 $\to$ 0.646 & 0.879 $\to$ 0.901 \\
    % Toblend & 7.03 $\to$ - & 0.378 $\to$ - & 0.753 $\to$ - & 6.88 $\to$ - & 0.558 $\to$ - & 0.724 $\to$ - & 6.59 $\to$ - & 0.481 $\to$ - & 0.731 $\to$ - \\
    PromptAttack & 12.46 $\to$ 13.17 & 0.689 $\to$ 0.688 & 0.928 $\to$ 0.930 & 11.61 $\to$ 12.26 & 0.690 $\to$ 0.680 & 0.930 $\to$ 0.941 & 10.39 $\to$ 10.51 & 0.653 $\to$ 0.648 & 0.934 $\to$ 0.941 \\
    Charmer & 14.54 $\to$ 13.65 & 0.550 $\to$ 0.545 & 0.971 $\to$ 0.967 & 12.88 $\to$ 11.83 & 0.641 $\to$ 0.599 & 0.974 $\to$ 0.966 & 11.62 $\to$ 11.57 & 0.544 $\to$ 0.530 & 0.987 $\to$ 0.983 \\
    DPO-Evader & 9.82 $\to$ 9.91 & 0.441 $\to$ 0.441 & 0.921 $\to$ 0.921 & 6.80 $\to$ 7.12 & 0.353 $\to$ 0.366 & 0.927 $\to$ 0.928 & 8.37 $\to$ 8.97 & 0.540 $\to$ 0.534 & 0.933 $\to$ 0.934 \\
    HMGC & 48.02 $\to$ 37.54 & 0.518 $\to$ 0.556 & 0.916 $\to$ 0.930 & 58.85 $\to$ 28.65 & 0.571 $\to$ 0.633 & 0.917 $\to$ 0.950 & 42.69 $\to$ 33.06 & 0.528 $\to$ 0.557 & 0.926 $\to$ 0.937 \\
    GradEscape & 21.51 $\to$ 17.82 & 0.615 $\to$ 0.635 & 0.958 $\to$ 0.963 & 14.69 $\to$ 12.31 & 0.702 $\to$ 0.676 & 0.967 $\to$ 0.970 & 20.40 $\to$ 15.98 & 0.616 $\to$ 0.620 & 0.950 $\to$ 0.958 \\
    CoPA & 27.52 $\to$ 21.76 & 0.639 $\to$ 0.667 & 0.873 $\to$ 0.904 & 29.95 $\to$ 17.80 & 0.635 $\to$ 0.696 & 0.874 $\to$ 0.923 & 23.75 $\to$ 18.40 & 0.645 $\to$ 0.681 & 0.870 $\to$ 0.901 \\
    \bottomrule
  \end{tabular}%
  }
    \caption{Performance comparison of baselines before and after applying the GPT-5 refinement model (Stage 4 of MASH). Metrics shown are \textbf{PPL} $\downarrow$ (lower is better), \textbf{GRUEN} $\uparrow$ (higher is better), and \textbf{BS} (BertScore) $\uparrow$ (higher is better). Format: \textit{Original} $\rightarrow$ \textit{Refined}.}
  \label{tab:refinement_comparison}
  
\end{table*}

\noindent \textbf{Generalizability of Refinement Module.} Table~\ref{tab:refinement_comparison} evaluates whether the inference-time refinement module (Stage 4) generalizes beyond MASH to other baseline attack methods. We apply the refinement module to outputs generated by ten baseline attacks across six datasets. The results demonstrate consistent improvements in text quality: PPL decreases substantially (e.g., TextFooler on Social: 136.89 → 47.22), and GRUEN scores improve across most methods. BERTScore remains stable or improves slightly, indicating that semantic consistency is preserved during refinement.

\section{Extended Analysis on Transferability}
\label{app:data_analysis}

This section analyzes the transferability anomalies observed in 
Section~\ref{sec:uni_and_trans}. While MASH generalizes effectively 
across creative and argumentative domains (e.g., Essay), it fails to 
transfer to Reuters. We attribute this to stylistic mismatch: the 
rhetorical diversity learned by MASH conflicts with the rigid, objective, 
and concise standards of news reportage, rendering adversarial samples 
easily detectable. This mismatch also explains why generic corpora 
(dmitva, SemEval) failed to bypass Reuters-trained detectors, as they 
likely misclassify non-journalistic styles as AI-generated; notably, 
the domain-specific cc\_news proved effective.

Regarding detector transferability (Figure~\ref{fig:detector_transfer}), 
attacks transfer from supervised to zero-shot methods, but not vice versa. 
Supervised detectors rely on high-dimensional semantic features, compelling 
MASH to perform deep stylistic alignment. Since zero-shot detectors utilize 
simpler statistical proxies inherent to human writing, samples satisfying 
the rigorous semantic constraints of RoBERTa naturally meet the statistical 
thresholds of Binoculars. Conversely, optimizing solely against zero-shot 
detectors often results in metric overfitting—merely reducing perplexity 
to cross numerical thresholds without achieving the stylistic coherence 
required to deceive supervised classifiers.

\section{Real Cases}
\label{sec:appendix_examples}

We randomly sample examples from our experiments, showing the original 
AI-generated texts alongside their MASH-polished versions to 
illustrate the stylistic transformation.

\begin{figure*}[t]
  \centering  % 图片居中
  \includegraphics[width=1\linewidth]{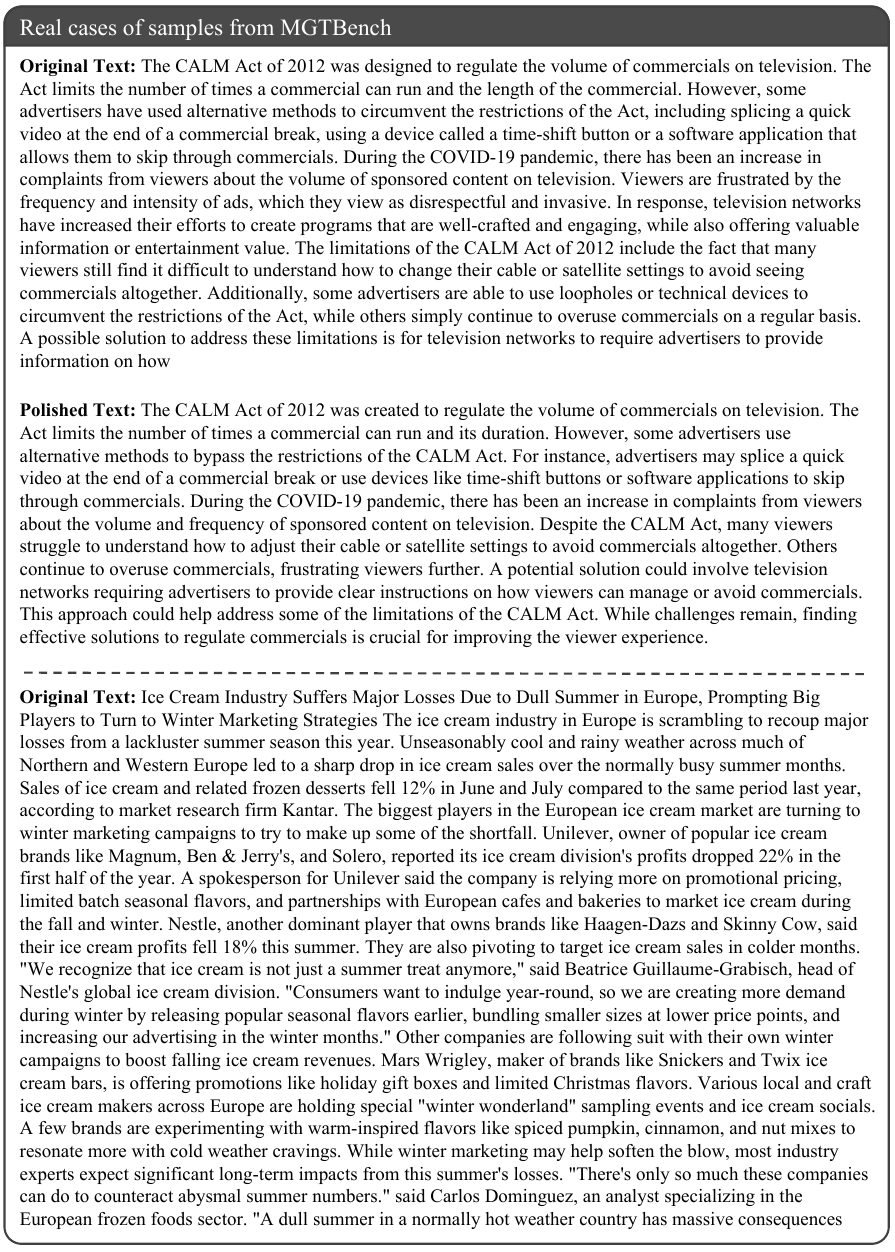} 
  \label{fig:real_case1} %以此标签引用图片
\end{figure*}

\begin{figure*}[t]
  \centering  % 图片居中
  \includegraphics[width=1\linewidth]{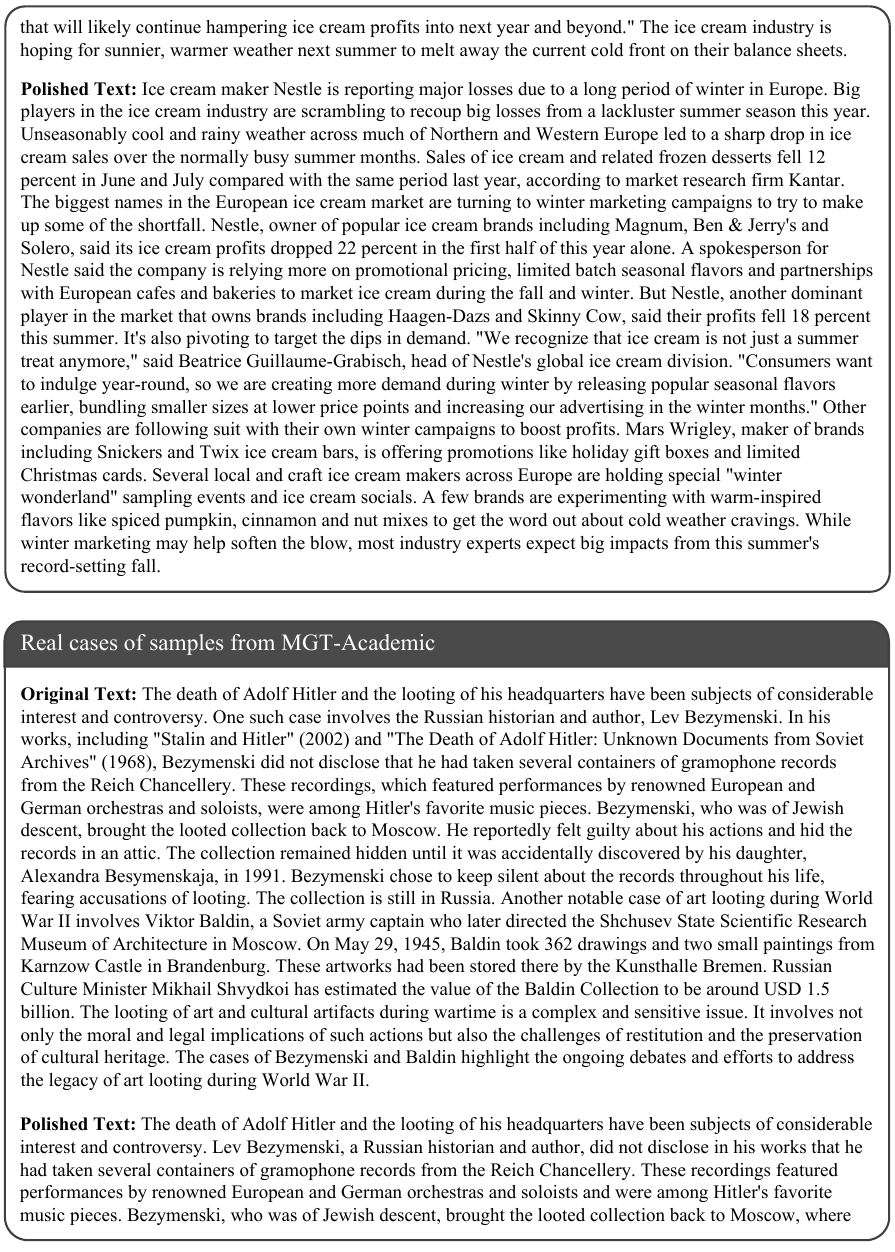} 
  \label{fig:real_case2} %以此标签引用图片
\end{figure*}

\begin{figure*}[t]
  \centering  % 图片居中
  \includegraphics[width=1\linewidth]{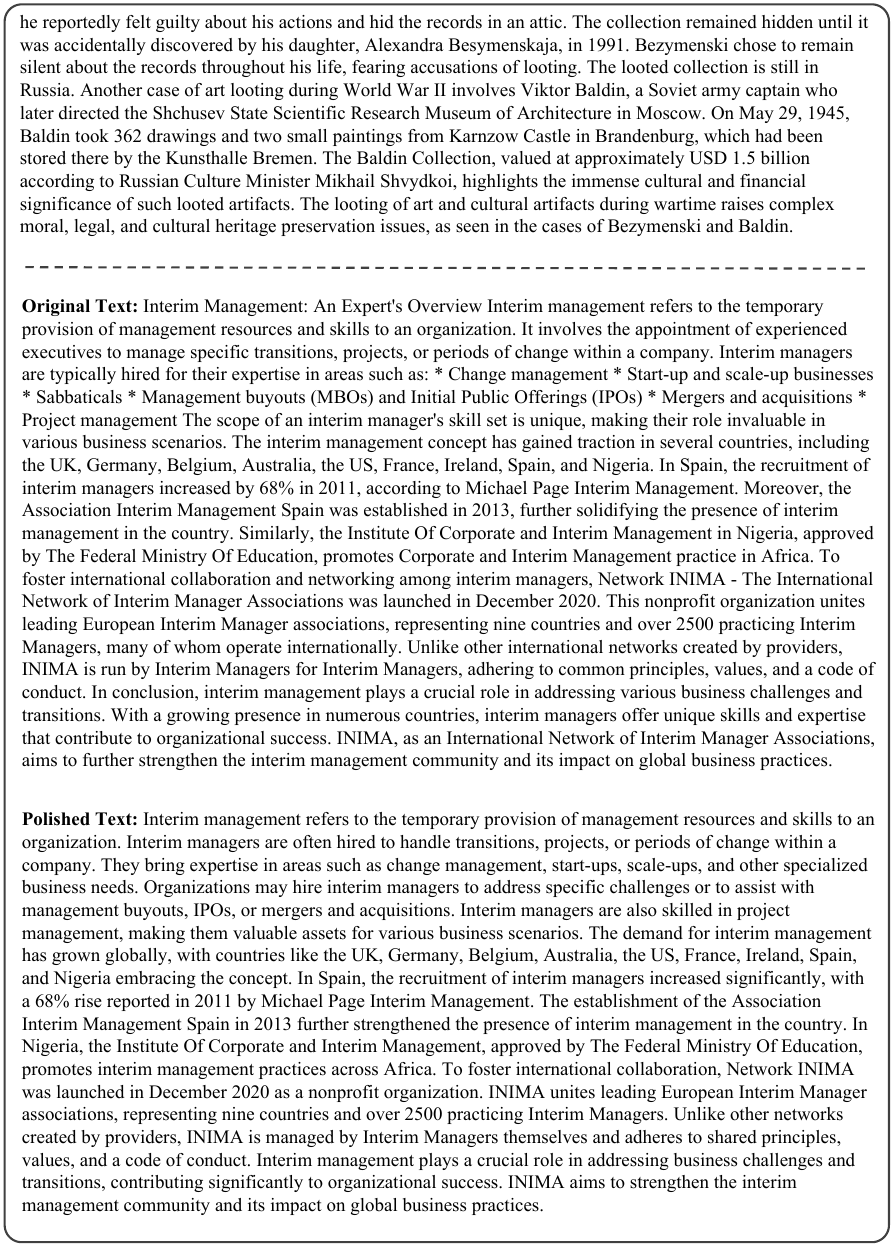} 
  \label{fig:real_case3} %以此标签引用图片
\end{figure*}

\end{document}